\documentclass[journal,final,twocolumn,10pt,twoside]{IEEEtranTCOM}

\normalsize

\ifCLASSINFOpdf
 
\else
 
\fi

\usepackage{mathtools}
\usepackage{amsmath}
\usepackage{amssymb}
\usepackage{amsmath}
\usepackage{cite}
\usepackage{algorithm, algorithmic}

\usepackage{siunitx}
\usepackage{multirow}
\usepackage{multicol}

\usepackage{graphicx}
\usepackage{graphicx}
\date{}

\usepackage{calc}
\newcolumntype{M}[1]{>{\centering\arraybackslash}m{#1}}
\newcolumntype{N}{@{}m{0pt}@{}}

\usepackage{mdwmath}
\usepackage{blindtext}
\usepackage{eqparbox}
\usepackage{fixltx2e}
\usepackage{stfloats}
\DeclareMathOperator{\E}{\mathbb{E}}

\usepackage{amsthm}
\usepackage{color}
\usepackage{xcolor}

\newtheorem{theorem}{{Theorem}}
\newtheorem{lemma}[theorem]{{Lemma}}

\newtheorem{corollary}[theorem]{{Corollary}}

\DeclareMathAlphabet{\mathbfsl}{OT1}{ppl}{b}{it} 

\def\QEDclosed{\mbox{\rule[0pt]{1.3ex}{1.3ex}}} 

\def\QED{\QEDclosed} 

\def\endproof{\hspace*{\fill}~\QED\par\endtrivlist\unskip}
\def\endproofempt{\hspace*{\fill}~$\square$\par\endtrivlist\unskip}

\newcommand{\be}[1]{\begin{equation}\label{#1}}
\newcommand{\ee}{\end{equation}}

\renewcommand{\leq}{\leqslant}
 
\renewcommand{\geq}{\geqslant}

\newcommand{\Tref}[1]{Theo\-rem\,\ref{#1}}

\newcommand{\Lref}[1]{Lem\-ma\,\ref{#1}}
\newcommand{\Cref}[1]{Co\-ro\-lla\-ry\,\ref{#1}}

\begin{document}
\title{Massive Coded-NOMA for Low-Capacity Channels:\\
	A Low-Complexity Recursive Approach}
\author{Mohammad Vahid Jamali and Hessam Mahdavifar,~\IEEEmembership{Member,~IEEE}
\thanks{The material in this paper was presented in part at the IEEE Global Communications Conference (GLOBECOM), Abu Dhabi, UAE, Dec. 2018 \cite{jamali2018low}.}
\thanks{The authors are with the Department of Electrical Engineering and Computer Science, University of Michigan, Ann Arbor, MI 48109, USA (e-mail: mvjamali@umich.edu, hessam@umich.edu).}
\thanks{This work was supported by the National Science Foundation
under grants CCF--1763348, CCF--1909771, and CCF--1941633.}
}
\maketitle
\begin{abstract}
In this paper, we present a low-complexity recursive approach for massive and scalable code-domain nonorthogonal multiple access (NOMA) with applications to emerging low-capacity scenarios. The problem definition in this paper is inspired by three major requirements of the next generations of wireless networks. Firstly, the proposed scheme is particularly beneficial in low-capacity regimes which is important in practical scenarios of utmost interest such as the Internet-of-Things (IoT) and massive machine-type communication (mMTC). Secondly, we employ code-domain NOMA to efficiently share the scarce common resources among the users. Finally, the proposed recursive approach enables code-domain NOMA with low-complexity detection algorithms that are scalable with the number of users to satisfy the requirements of massive connectivity. To this end, we propose a novel encoding and decoding scheme for code-domain NOMA based on factorizing the pattern matrix, for assigning the available resource elements to the users, as the Kronecker product of several smaller factor matrices. As a result, both the pattern matrix design at the transmitter side and the mixed symbols' detection at the receiver side can be performed over matrices with dimensions that are much smaller than the overall pattern matrix. Consequently, this leads to significant reduction in both the complexity and the latency of the detection. We present the detection algorithm for the general case of factor matrices. The proposed algorithm involves several recursions each involving certain sets of equations corresponding to a certain factor matrix. We then characterize the system performance in terms of average sum rate, latency, and detection complexity. Our latency and complexity analysis confirm the superiority of our proposed scheme in enabling large pattern matrices. Moreover, our numerical results for the average sum rate show that the proposed scheme provides better performance compared to straightforward code-domain NOMA with comparable complexity, especially at low-capacity regimes.
\end{abstract}
\IEEEpeerreviewmaketitle
\section{Introduction}\label{Sec1}
\IEEEPARstart{L}ow-capacity scenarios have become increasingly important in a variety of emerging applications such as the Internet-of-Things (IoT) and massive machine-type communication (mMTC) \cite{fereydounian2019channel}.
 For example, the narrowband IoT (NB-IoT) feature, included in Release-13 of the 3rd generation partnership project (3GPP), is specifically meant for ultra-low-rate, wide-area, and low-power applications \cite{wang2017primer}.
To ensure wide-area applications, NB-IoT is designed to support maximum coupling losses (MCLs) as large as $170$ \si{dB}. Achieving such large MCLs requires reliable detection for signal-to-noise ratios (SNRs) as low as $-13$ \si{dB} \cite{ratasuk2016overview}. 
Consequently, one needs to carefully design the communication protocols aimed for massive communication applications, such as IoT and mMTC,
with respect to the low-SNR constraints.

Recently, nonorthogonal multiple access (NOMA) techniques, that borrow ideas from solutions to traditional problems in network information theory including multiple access and broadcast channels, have gained significant attention. In NOMA schemes, multiple users are served in the same orthogonal resource element (RE) or, more generally speaking, a set of users are served in a smaller set of REs. The goal is to significantly increase the system throughput and reliability, improve the users' fairness, reduce the latency, and support massive connectivity \cite{ding2017survey,wan2019promising}. NOMA is a general setup and, in principle, any multiple access scheme that attempts to \textit{non-orthogonally} share the REs among the users, e.g., random multiple access \cite{ployISIT2017_1,mag2018,yuan2020iterative} and opportunistic approaches \cite{shahsavari2018opportunistic}, can be formulated in this setting. In general, NOMA techniques in the literature can be classified into two categories: power-domain NOMA and code-domain NOMA.

Power-domain NOMA serves multiple users in the same orthogonal RE by properly allocating different power levels to the users \cite{islam2016power}. Power-domain NOMA has attracted significant attention in recent years and several problems have been explored in this context. This includes cooperative communication \cite{ding2015cooperative,li2019residual,pei2020secure}, simultaneous wireless information and power transfer (SWIPT) \cite{liu2016cooperative}, multiple-input multiple-output (MIMO) systems \cite{sun2015ergodic}, mmWave communications \cite{ding2017random}, mixed radio frequency and free-space optics (RF-FSO) systems \cite{jamali2018outage,jamali2019uplink2}, unmanned aerial vehicles (UAVs) communications \cite{hou2019exploiting}, and cache-aided systems \cite{doan2019power}.
The detection procedure in power-domain NOMA highly relies on successive interference cancellation (SIC) which works well provided that the channel conditions of the users paired together are not close to each other.

Code-domain NOMA, on the other hand, aims at serving a set of users, say $K$, in a set of $M$ orthogonal REs, with $M\leq K$, using a certain code/pattern matrix. The pattern matrix comprises $K$ pattern vectors each assigned to a user specifying the set of available REs to that user. Unlike power-domain NOMA, code-domain NOMA works well even in power-balanced scenarios provided that each user is served by a unique pattern vector. Nevertheless, code-domain NOMA has received relatively less attention in the literature compared to the power-domain NOMA. This is mainly because its gain usually comes at the expense of complex multiuser detection (MUD) algorithms, such as maximum \textit{a posteriori} (MAP) detection, {message passing algorithm (MPA)}, and maximum likelihood (ML) detection. 
In this context,
sparse code multiple access (SCMA) is proposed in \cite{nikopour2013sparse} based on directly mapping the incoming bits to multidimensional codewords of a certain SCMA codebook set. Also, a low-complexity SCMA decoding algorithm is proposed in \cite{wei2016low2} based on list sphere decoding. Another efficient SCMA decoder is discussed in \cite{zhang2020efficient} based on deterministic message passing algorithms. Lattice partition multiple access (LPMA) is proposed in \cite{fang2016lattice,qiu2018lattice1,qiu2018lattice2} based on multilevel lattice codes for multiple users. Interleave-grid multiple access (IGMA) is proposed in \cite{xiong2017advanced,hu2018nonorthogonal} in order to increase the user multiplexing capability and to improve the performance.
Moreover, pattern division multiple access (PDMA) is introduced in \cite{dai2014successive,chen2017pattern}, where the pattern vectors are designed with disparate orders of transmission diversity to mitigate the error propagation problem in SIC receivers. As opposed to most of the aforementioned prior works in code-domain NOMA, the pattern matrix in PDMA is not necessarily sparse. In other words, the number of REs allocated to a particular user can potentially be comparable to the total number of REs. { Very recently, a low-complexity on-off division multiple access (ODMA) scheme has been proposed in \cite{song2020super} for  code-domain NOMA systems. In the ODMA scheme proposed in \cite{song2020super}, each user employs the same channel code whose coded bits, after modulation, are sent in a random time-hopping manner over much larger number of time slots than the length of the channel code by keeping the remaining time slots idle, resulting in a super-sparse multiple access with a very low-complexity iterative multi-user decoding method. Also,  a coded MIMO-NOMA system with capacity-approaching performance and low implementation complexity has been proposed in \cite{chi2018practical} which consists of a linear minimum mean-square error (LMMSE) multi-user detector and a bank of single-user message-passing decoders to decompose the overall NOMA signal recovery into distributed low-complexity computations with iterative processing. Additionally, the achievable rates of approximate message passing (AMP) algorithm for coded random linear systems has been analyzed in \cite{liu2019capacity}, proving that the low-complexity AMP algorithm achieves the constrained capacity based on matched forward error control (FEC) coding.}

Design of the pattern matrix plays a critical role in code-domain NOMA to balance the trade-off between the system performance and the complexity. Impact of the pattern matrix on the average sum-rate of SCMA systems is explored in \cite{yang2016impact}, where a low-complexity iterative algorithm is proposed to facilitate the design of the pattern matrix. Moreover, the total throughput of low-density code-domain (LDCD) NOMA is characterized in \cite{shental2017low} for regular random pattern matrices with large dimensions. It is well understood that, for a given overload factor $\beta\triangleq K/M$, expanding the dimension of the pattern matrix improves the system performance \cite{chen2017pattern}.  However, increasing the pattern matrix dimension significantly increases the detection complexity.

Inspired by the aforementioned trade-off between the system performance and the detection complexity, we propose a novel encoding and decoding approach toward code-domain NOMA which factorizes the pattern matrix as the Kronecker product of several smaller factor matrices. Consequently, as we show, both the pattern matrix design at the transmitter side and the mixed symbols' detection at the receiver side can be performed over much smaller dimensions and with significantly reduced complexity, through recursive detection.
{In other words, our scheme enables application of pattern matrices with large dimensions while keeping the overall detection complexity and latency manageably low. As we establish later, this provides the possibility to significantly improve the system performance at a given, reasonably low, complexity and latency level.}

 The low complexity of the proposed detection algorithm for pattern matrices with large dimensions enables grouping a massive number of users together, referred to as ``massive coded-NOMA'' in this paper. Moreover, the proposed multiple access technique is particularly advantageous in scenarios where there is a stringent constraint on the maximum power of the users' symbols. Therefore, even full-power transmission of the users' symbols (up to a predefined maximum allowed power) over only one RE (or few REs) does not meet the users' desired data rates. In such circumstances, users are inevitable to transmit over several REs, possibly with the maximum allowed power per symbol, to achieve the desired data rates. For instance, a large number of repetitions are allowed in low-capacity scenarios such as NB-IoT and mMTC in order to enable reliable communications in these emerging applications \cite{ratasuk2016overview}.
Our proposed protocol, through facilitating coded-NOMA over large-dimension pattern matrices, enables spreading the symbols of the low-capacity users over several REs and then detecting them with a low complexity. { We start our studies with the traditional Gaussian multiple access channel (GMAC) model, and then clarify how the results can be extended to more practical scenarios such as fading channels.}

{ The main contributions of the paper are summarized as follows.
	\begin{itemize}
		\item We propose a low-complexity and scalable approach toward code-domain NOMA by constructing the overall pattern matrix as the Kronecker product of several factor matrices. We show that the proposed scheme significantly reduces the detection complexity and also facilitates the design of good pattern matrices; hence, it allows incorporating large pattern matrices that are of particular interest for massive communication and low-capacity channels.
		\item For the Kronecker product of square factor matrices we propose a systematic way of choosing the factor matrices that enables a remarkably low-complexity detection algorithm involving only few linear operations at each recursion. We show that our design scheme and the proposed detection algorithm for the Kronecker product of square factor matrices effectively increases the SNRs of data symbols after each recursion that is of particular importance for low-capacity channels.
		\item We provide a generic recursive detection algorithm for the Kronecker product of rectangular factor matrices that can work on the general case of factor matrices.
		\item We derive useful expressions for the characterization of important system performance metrics such as the average sum-rate per RE, latency, and detection complexity. 
		\item We demonstrate how the results of the paper, derived for the Gaussian MAC model, can be extended to practical scenarios such as uplink and downlink fading channels. We also highlight how the proposed code-domain NOMA scheme in this paper can be combined with power-domain NOMA to boost the system performance in various aspects.
		\item We provide extensive numerical analysis to study the system performance in various scenarios.
\end{itemize}}

The rest of the paper is organized as follows. In Section \ref{Sec2}, we briefly describe the system and channel models. In Section \ref{Sec3}, we focus on the design procedure and detection algorithm of the proposed protocol over GMAC. We then characterize the average sum rate, latency, and detection complexity in Section \ref{Sec4}. 
 We devote Section \ref{Sec5} to more realistic scenarios, and Section \ref{Sec6} to numerical results. Finally, we conclude the paper and highlight several future research directions in Section \ref{Sec7}.

\section{Preliminaries and System Model}\label{Sec2}

In this section, we briefly review the basics of code-domain NOMA, specifically PDMA \cite{chen2017pattern}, which is relevant to the system model of our proposed scheme. 
We consider a collection of $K$ users communicating using $M$ REs, and define the overload factor as $\beta\triangleq K/M\geq 1$ implying the multiplexing gain of PDMA compared to the orthogonal multiple access (OMA).
The modulation symbol $x_k$ of the $k$-th user is spread over $M$ orthogonal REs using the pattern vector $\boldsymbol{g}_k$ as $\boldsymbol{v}_k\triangleq\boldsymbol{g}_k x_k$, $1\leq k\leq K$, where $\boldsymbol{g}_k\in{\mathcal{B}^{M\times 1}}$, with $\mathcal{B}\triangleq\{0,1\}$, is an $M\times 1$ binary vector defining the set of REs available to the $k$-th user; the $k$-th user can use the $i$-th RE, $1\leq i\leq M$, if the $i$-th element of $\boldsymbol{g}_k$ is ``$1$'', i.e., if ${g}_{i,k}\triangleq\boldsymbol{g}_{k}(i)=1$. Otherwise, if ${g}_{i,k}=0$, then the $k$-th user does not use the $i$-th RE. Then the overall ($M\times K$)-dimensional pattern matrix $\boldsymbol{G}\in{\mathcal{B}^{M\times K}}$ is specified as follows \cite{chen2017pattern}:
\begin{align}\label{eq1}
	\boldsymbol{G}_{M\times K}\triangleq
	\begin{bmatrix}
		\boldsymbol{g}_{1} & \boldsymbol{g}_{2}& \cdots &\boldsymbol{g}_{K}
	\end{bmatrix}=
	\begin{bmatrix}
		{g}_{i,k}
	\end{bmatrix}_{{M\times K}}.
\end{align}

For the uplink transmission each user transmits its spread symbol to the base station (BS). Therefore, assuming perfect synchronization at the BS, the
vector $\boldsymbol{y}$ comprising the received signals at all $M$ REs can be modeled as
\begin{align}\label{eq2}
	\boldsymbol{y}=\sum\nolimits_{k=1}^{K}{\rm diag}(\boldsymbol{h}_k)\boldsymbol{v}_k+\boldsymbol{n},
\end{align}
where $\boldsymbol{h}_k$ is the vector modeling the uplink channel response of the $k$-th user at all of $M$ REs, and $\boldsymbol{n}$ is the noise vector at the BS with length $M$. Furthermore, ${\rm diag}(\boldsymbol{h}_k)$ is a diagonal matrix consisting of the elements of $\boldsymbol{h}_k$. Then the uplink transmission model can be reformulated as
\begin{align}\label{eq3}
	\boldsymbol{y}=\boldsymbol{H}\boldsymbol{x}+\boldsymbol{n},
\end{align}
where $\boldsymbol{x}\triangleq[x_1,x_2,\dots,x_K]^T$, $\boldsymbol{H}\triangleq\boldsymbol{\mathcal{H}}{\odot}\boldsymbol{G}_{M\times K}$ is the PDMA equivalent uplink channel response, $\boldsymbol{\mathcal{H}}\triangleq[\boldsymbol{h}_1,\boldsymbol{h}_2,\dots,\boldsymbol{h}_K]$, $\boldsymbol{A}^T$ is the transpose of the matrix $\boldsymbol{A}$, and $\odot$ denotes the element-wise product \cite{chen2017pattern}.

Moreover, for the downlink transmission, the BS first encodes the data symbol of each user according to its pattern vector and then transmits the superimposed encoded symbols $\sum_{j=1}^{K}\boldsymbol{v}_j$ through the channel. Therefore, the received signal $\boldsymbol{y}_k$ at the $k$-th user can be expressed as
\begin{align}\label{eq4}
	\boldsymbol{y}_k&={\rm diag}(\boldsymbol{h}_k)\sum\nolimits_{j=1}^{K}\boldsymbol{g}_jx_j+\boldsymbol{n}_k
	=\boldsymbol{H}_k\boldsymbol{x}+\boldsymbol{n}_k,
\end{align}
where $\boldsymbol{H}_k\triangleq {\rm diag}(\boldsymbol{h}_k)\boldsymbol{G}_{M\times K}$ is the PDMA equivalent downlink channel response of the $k$-th user. Moreover, $\boldsymbol{h}_k$ and $\boldsymbol{n}_k$ are the downlink channel response and the noise vector, both with length $M$, at the $k$-th user, respectively \cite{chen2017pattern}.

As a standard MUD algorithm in NOMA systems SIC provides a proper trade-off between the system performance and the complexity. However, SIC receivers often suffer from error propagation problems as the system performance is highly dependent on the correctness of early-detected symbols. In order to resolve this issue, one can either improve the reliability of the initially-decoded users or employ more advanced detection algorithms including ML and MAP. In \cite{chen2017pattern}, disparate diversity orders are adopted for different users by assigning patterns with heavier weights to those early-detected users in order to increase their transmission reliability. Furthermore, employing more advanced detection algorithms severely increases the system complexity especially for larger pattern matrices; this may hinder their practical implementation particularly for downlink transmission where the users are supposed to have lower computational resources than the BS.

 In the next section, we elaborate how the proposed scheme scales with the dimension of the pattern matrix, even with rather heavy pattern weights, to boost massive connectivity without a significant increase on the overall system complexity. In Sections \ref{Sec3} and \ref{Sec4}, we consider the conventional model of GMAC, i.e., $\boldsymbol{y}=\boldsymbol{G}\boldsymbol{x}+\boldsymbol{n}$, that corresponds to the case where
{ $\boldsymbol{\mathcal{H}}$ defined after \eqref{eq3} is replaced by an all-one matrix (or $\boldsymbol{h}_k$ defined after \eqref{eq4} is replaced by an all-one vector). The simple channel model of GMAC, although is conventional in the literature, has several limitations in real practice, e.g., it ignores the information about the channel gains. However, since this is an initial research on this topic, we start with the simplest channel model to streamline the presentation of the underlying design strategies and detection algorithms for the proposed low-complexity code-domain NOMA.
 We then demonstrate, in Section \ref{Sec5A}, how the results can be extended to the cases of uplink and downlink fading channels, as in \eqref{eq3} and \eqref{eq4}, respectively.}

\section{Design and Detection over Gaussian MAC}\label{Sec3}
\subsection{Design Principles}\label{Sec3A}
Depending on the number of users and the available REs to them, we propose to consider pattern matrices that are factorized as the
Kronecker product of a certain number, denoted by $L$, of smaller factor matrices as follows:
\begin{align}\label{eq5}
	\boldsymbol{G}_{M\times K}=\boldsymbol{G}^{(1)}_{m_1\times k_1}\otimes \boldsymbol{G}^{(2)}_{m_2\times k_2} \otimes \dots \otimes \boldsymbol{G}^{(L)}_{m_L\times k_L},
\end{align}
where $L$ is a design parameter and  $\otimes$ denotes the Kronecker product defined as
\begin{align}\label{eq6}
	\boldsymbol{A}_{m\times k}\otimes {\boldsymbol{B}}
	\triangleq\begin{bmatrix}
		a_{11}\boldsymbol{B} & \cdots & a_{1k}\boldsymbol{B} \\ 
		\vdots & \ddots  & \vdots \\ 
		a_{m1}\boldsymbol{B} & \cdots  & a_{mk}\boldsymbol{B}
	\end{bmatrix},
\end{align} 
for any two matrices $\boldsymbol{A}$ and $\boldsymbol{B}$. The dimensions of the resulting pattern matrix in \eqref{eq5} relate to the dimensions of the factor matrices as $M=\prod_{l=1}^{L}m_l$ and $K=\prod_{l=1}^{L}k_l$. In general, if at least one of $M$ or $K$ is not a prime number, we can find some $m_l>1$ or $k_l>1$ to construct the pattern matrix as the Kronecker product of some smaller factor matrices. On the other hand, if both $M$ and $K$ are prime numbers, $L$ is equal to one and the design procedure simplifies to that of  conventional pattern matrix design (such as PDMA \cite{chen2017pattern}). Note that both $M$ and $K$ are design parameters and we can always properly group a certain number of users over a desired number of REs to optimize the system performance and the complexity. 

{ The main inspiration behind factorizing the overall pattern matrix as the Kronecker product of smaller factor matrices is to come up with a low-complexity recursive detection algorithm that can be applied to large-scale code-domain NOMA settings. As it will be elaborated in the next subsections, the proposed structure enables a recursive detection algorithm that at each recursion divides the set of equations to disjoint subsets of equations and works on them in parallel. In that sense, the factorization helps us to devise a ``\textit{divide-and-conquer}'' type of algorithm that recursively breaks down the detection problem into sub-problems that can be executed in parallel. This then significantly lowers the complexity and latency compared to the case where we directly solve the system of equations defined according to the overall pattern matrix.}

The proposed structure not only alleviates the detection complexity and latency at the receiver side but also significantly reduces the search space at the transmitting party, enabling the usage of large pattern matrices with a reasonable complexity for massive connectivity. {With the search space, we mean the number of possible pattern matrices for which one should perform a (possibly brute-force) search over them to obtain a \textit{good} pattern matrix. The notion of a \textit{good} pattern matrix depends on the context and the objective function for which we want to optimize. With that being said, a regular pattern matrix design (without factorization) requires a comprehensive search over all $\binom{{2^M}-1}{K}$ possible ($M\times K$)-dimensional binary matrices with distinct nonzero columns (patterns assigned to each user) to find an optimal pattern matrix \cite{chen2017pattern}. 
On the other hand, it is easy to show that satisfying distinct nonzero columns for the overall pattern matrix $\boldsymbol{G}_{M\times K}$ of the form given by \eqref{eq5} requires distinct nonzero columns for all of the factor matrices $\boldsymbol{G}^{(l)}_{m_l\times k_l}$, $l=1,2,\dots,L$.
Otherwise, if any of the factor matrices has a repeated column, many of the pattern vectors in the overall pattern matrix will be the same, i.e., many of the users will be served with a same pattern vector which itself requires more advanced detection algorithms to distinguish them at the receiver (see also Section \ref{Sec5B}). Therefore, the search space for our proposed design method reduces to $\prod_{l=1}^{L}\binom{2^{m_l}-1}{k_l}$ which is significantly smaller than $\binom{{2^M}-1}{K}$.}
For example, for $M=6$ and $K=9$, the regular pattern matrix design requires searching over $\binom{{2^6}-1}{9}=2.36\times 10^{10}$ possible matrices, while our design method with the factorization of $\boldsymbol{G}^{(1)}_{2\times 3} \otimes \boldsymbol{G}^{(2)}_{3\times 3}$ only needs to search over $\binom{{2^2}-1}{3}.\binom{{2^3}-1}{3}=35$ matrices. Note that recursive construction of large matrices based on the Kronecker product of some smaller matrices has been used in different contexts such as polar coding \cite{arikan2009channel} and its extended versions such as compound polar coding \cite{mahdavifar2013compound,mahdavifar2015polar}.

{In the following subsections, we first describe the detection algorithm over the Kronecker product of square and rectangular factor matrices, and then summarize the overall detection algorithm over the general case of the pattern matrix. It is worth mentioning at this point that, as it will be clarified from the extensive characterizations in the following subsections, our proposed detection algorithm for the case of square factor matrices significantly differs from that of the rectangular factor matrices and cannot be obtained as a special case of that.}
\subsection{Square Factor Matrices}\label{Sec3B}
In this subsection, we explore the design of square factor matrices and describe the corresponding detection algorithm over GMAC. In particular, we assume that the overall pattern matrix is represented as $\boldsymbol{G}_{M\times K}=\boldsymbol{P}^{(1)}_{m_1\times m_1}\otimes \boldsymbol{P}^{(2)}_{m_2\times m_2} \otimes \dots \otimes \boldsymbol{P}^{(L)}_{m_L\times m_L}$, where $\boldsymbol{P}^{(l)}_{m_l\times m_l}$, $l=1,2,\dots,L$, is an $m_l\times m_l$ binary square matrix. In this case, $M=K=\prod_{l=1}^{L}m_l$ and the overload factor $\beta$ is equal to $1$. However, we will observe that with a careful design of the square factor matrices one can improve the effective SNR of individual data symbols to a desired level that can guarantee a predetermined data rate. As it will be clarified in the next example, we define the effective SNR as the SNR of the individual data symbols at the end of the detection algorithm after several rounds of recursive combining. The aforementioned gain is obtained by a low-complexity detection algorithm involving only few linear operations (additions/subtractions) at each recursion, as detailed in Section \ref{Sec3B2}. Therefore, this scheme is useful especially when the transmission of each symbol over each RE is constrained by a maximum power limit which is often the case in low-capacity scenarios as discussed in Section \ref{Sec1}. To proceed, we begin with the following illustrative example that helps clarifying the design procedure and the recursive detection algorithm provided afterward. 

\noindent{\textbf{Example 1.}} Consider the  transmission of $K=12$ users (symbols) over $M=12$ REs realized using the Kronecker product of the following two square factor matrices
\vspace{-0.05in}
\begin{align}\label{eq7}
\boldsymbol{P}^{(1)}_{3\times 3}=
\begin{bmatrix}
1 & 1 & 0 \\ 
1 & 0 & 1\\ 
0 & 1  & 1
\end{bmatrix}, ~~~~ \boldsymbol{P}^{(2)}_{4\times 4}=
\begin{bmatrix}
0 & 0 & 0 & 1 \\
0 & 1 & 1 & 0\\ 
1 & 0 & 1  & 0\\
1 & 1 & 0 & 0
\end{bmatrix}.
\end{align}
Let us further define the following square matrices $\boldsymbol{\alpha}^{(l)}$, $l=1,2$, such that the matrix product $\boldsymbol{\alpha}^{(l)}\boldsymbol{P}^{(l)}$ is diagonal. As it will be clarified after this example and the detailed detection algorithm in Section \ref{Sec3B2}, combining the received signals according to the rows of matrices $\boldsymbol{\alpha}^{(l)}$'s significantly reduces the number of unknown data symbols after each recursion.
\begin{align}\label{eq8_alpha12}
\!\boldsymbol{\alpha}^{(1)}_{3\times 3}\!=\!\!
\begin{bmatrix}
1 & 1 & -1 \\ 
1 & -1 & 1\\ 
-1 & 1  & 1
\end{bmatrix}\!\!, ~\boldsymbol{\alpha}^{(2)}_{4\times 4}\!=\!\!
\begin{bmatrix}
0 & -1 & 1 & 1 \\
0 & 1 & -1 & 1\\ 
0 & 1 & 1  & -1\\
1 & 0 & 0 & 0
\end{bmatrix}\!\!.
\end{align}
{Then $\boldsymbol{y}_{12\times 1}=\boldsymbol{G}_{12\times 12}\boldsymbol{x}_{12\times 1}+\boldsymbol{n}_{12\times 1}$ defines the received signals vector over all $12$ REs, where $\boldsymbol{y}_{12\times 1}\triangleq [y_1,y_2,\dots,y_{12}]^T$, $\boldsymbol{G}_{12\times 12}=\boldsymbol{P}^{(1)}_{3\times3}\otimes\boldsymbol{P}^{(2)}_{4\times4}$,  $\boldsymbol{x}_{12\times 1}\triangleq [x_1,x_2,\dots,x_{12}]^T$, and $\boldsymbol{n}_{12\times 1}\triangleq [n_1,n_2,\dots,n_{12}]^T$. Using the definition of the Kronecker product in \eqref{eq6}, we have, for example, for the received signals over the first four REs $y_1=x_4+x_8+n_1$, $y_2=x_2+x_3+x_6+x_7+n_2$, $y_3=x_1+x_3+x_5+x_7+n_3$, and $y_4=x_1+x_2+x_5+x_6+n_4$.}

The resulting set of $12$ equations for the received signals $y_i$'s over all REs can be analyzed using different MUD methods to detect the data symbols $x_{i}$'s, $i=1,2,\dots,12$. However, the proposed design method in this paper enables recursive detection of the data symbols with a significantly lower complexity. Here, for the sake of brevity, we focus on the detection of $x_1$, $x_5$, and $x_9$. The rest of the symbols can be detected in a similar fashion.
Let us define $y^{(1)}_1$, $y^{(1)}_5$, and $y^{(1)}_9$ each as the linear combination of four consecutive received signals according to the first row of $\boldsymbol{\alpha}^{(2)}$, i.e.,
\begin{align}\label{eq9}
y^{(1)}_1&\triangleq y_3+y_4-y_2=2x_1+2x_5+n^{(1)}_1,\nonumber\\
y^{(1)}_5&\triangleq y_7+y_8-y_6=2x_1+2x_9+n^{(1)}_5,\nonumber\\
y^{(1)}_9&\triangleq y_{11}+y_{12}-y_{10}=2x_5+2x_9+n^{(1)}_9,
\end{align}
where $n^{(1)}_1\triangleq n_3+n_4-n_2$, $n^{(1)}_5\triangleq n_7+n_8-n_6$, and $n^{(1)}_9\triangleq n_{11}+n_{12}-n_{10}$. In a given equation involving a certain number of data symbols and a noise term, we refer to the ratio of the power of each data symbol to the noise variance as the effective SNR of that data symbol. For instance, the effective SNR of each of the symbols $x_1$, $x_5$, and $x_9$ in \eqref{eq9} is $4/3$ times the original SNR. This is because the noise terms $n_{i}$'s in different REs are independent, i.e., each of $n^{(1)}_j$'s, for $j=1,5,9$, has a mean equal to zero and variance $\sigma_1^2=3\sigma^2$, where $\sigma^2$ is the variance of the original noise terms $n_{i}$'s. In this case, roughly speaking, we say that the effective SNR is increased by a factor of $4/3$ through this recursion{\footnote{{Note that, as it will be seen at the end of Example 1 and also will be demonstrated in Section \ref{Sec3B2} (see also Fig. \ref{tree1}), the proposed detection algorithm for the Kronecker product of $L$ square factor matrices reduces (at the end of the $L$-th recursion) to singleton equations each containing a single data symbol mixed with an additive noise term. Therefore, the notion of ``effective SNR", defined here, is more relevant than the common notion of signal-to-interference-plus-noise ratio (SINR).}}}. 

Note that the set of three equations in \eqref{eq9} is defined according to the factor matrix $\boldsymbol{P}^{(1)}_{3\times3}$ as follows:
\begin{align}\label{eq10_1}
\begin{bmatrix}
y^{(1)}_1\\y^{(1)}_5\\y^{(1)}_9
\end{bmatrix}=
\begin{bmatrix}
1 & 1 & 0 \\ 
1 & 0 & 1\\ 
0 & 1  & 1
\end{bmatrix}\begin{bmatrix}
2x_1\\2x_5\\2x_9
\end{bmatrix}+\begin{bmatrix}
n^{(1)}_1\\n^{(1)}_5\\n^{(1)}_9
\end{bmatrix}.
\end{align}
Therefore, combining the three new symbols
  according to the rows of $\boldsymbol{\alpha}^{(1)}$ results in the following set of equations involving only one data symbol, also referred to as singleton equations,
\begin{align}\label{eq10}
y^{(2)}_1&\triangleq y^{(1)}_1+y^{(1)}_5-y^{(1)}_9=4x_1+n^{(2)}_1,\nonumber\\
y^{(2)}_5&\triangleq y^{(1)}_1+y^{(1)}_9-y^{(1)}_5=4x_5+n^{(2)}_5,\nonumber\\
y^{(2)}_9&\triangleq y^{(1)}_5+y^{(1)}_9-y^{(1)}_1=4x_9+n^{(2)}_9,
\end{align}
in which $n^{(2)}_1\triangleq n^{(1)}_1+n^{(1)}_5-n^{(1)}_9$, $n^{(2)}_5\triangleq n^{(1)}_1+n^{(1)}_9-n^{(1)}_5$, and $n^{(2)}_9\triangleq n^{(1)}_5+n^{(1)}_9-n^{(1)}_1$. 
Hence, the effective SNR is increased again by a factor of $4/3$ since the noise components $n^{(2)}_j$'s have mean zero and variance $\sigma_2^2=3\sigma_1^2=3^2\sigma^2$ due to the independence of $n^{(1)}_j$'s. Finally, \eqref{eq10} can be used to decode the original data symbols $x_j$'s though with the effective SNRs increased by a factor of $(4/3)^2$. Applying a similar method demonstrates that the effective SNR of the data symbols $x_4$, $x_8$, and $x_{12}$ is improved by a factor of $4/3$ while the gain on all other data symbols is $(4/3)^2$ {(please refer to Section \ref{Sec3B3} and Fig. \ref{tree2} for the general characterization of the effective SNR development through the proposed recursive detection algorithm in Section \ref{Sec3B2})}.\endproofempt

The above example illustrates the idea behind the proposed structure for the square pattern matrix and shows its potential advantages by enabling a low-complexity recursive detection process. However, there are several important questions that need to be carefully addressed. In particular, what is the criteria for selecting the factor matrices? How can the received signals in different REs be combined to get smaller dimensions and simpler sets of equations (e.g., converting the original equations for $y_i$'s to \eqref{eq9} and then \eqref{eq9} to \eqref{eq10})? What exactly will be the gain of such a combining in terms of increasing the effective SNRs? And, how many equations and with what dimensions will be left at the end to perform advanced detection algorithms? Next, we aim at properly answering these questions with respect to a square pattern matrix factorized as the Kronecker product of smaller square factor matrices.

\subsubsection{Pattern Matrix Design and Recursive Combining}\label{Sec3B1}
With the pattern matrix structure, $\boldsymbol{G}_{M\times K}=\boldsymbol{P}^{(1)}_{m_1\times m_1}\otimes \boldsymbol{P}^{(2)}_{m_2\times m_2} \otimes \dots \otimes \boldsymbol{P}^{(L)}_{m_L\times m_L}$, both the combining procedure of the received signals over different REs and the resulting SNR gains, in each recursion, directly relate to the underlying square factor matrices $\boldsymbol{P}^{(l)}$'s. As it will be elaborated in Section \ref{Sec3B2}, our proposed detection algorithm starts from the rightmost factor matrix $\boldsymbol{P}^{(L)}$ and involves $L$ recursions. We will further establish that the $l'$-th recursion, for $l'=1,2,\dots,L$, involves equations defined according to the factor matrix $\boldsymbol{P}^{(l)}$ with $l=L-l'+1$.
 Now,
Let $\boldsymbol{P}^{(l,v)}$, for $v=1,2,\dots,\binom{2^{m_l}-1}{m_l}$, denote the $v$-th possible matrix for $\boldsymbol{P}^{(l)}$. Then at the beginning of the $l'$-th recursion we have certain sets of auxiliary equations of the following form (see, e.g., \eqref{eq10_1}):
\begin{align}\label{eq13}
	\!\!\begin{bmatrix}
		y^{(l'\!-\!1)}_{i_1}\\
		\vspace{-0.32cm}\\
	y^{(l'\!-\!1)}_{i_2}\\
		\vdots\\
y^{(l'\!-\!1)}_{i_{m_l}}
	\!\!\end{bmatrix}\!\!=\!\!
	\begin{bmatrix}
		p^{(l,v)}_{1,1}\!\! & \!\!p^{(l,v)}_{1,2} \!\!&\!\! \cdots \!\!&\!\! p^{(l,v)}_{1,m_l}\\
		\vspace{-0.25cm}\\
		p^{(l,v)}_{2,1} \!\!&\!\! p^{(l,v)}_{2,2} \!\!&\!\! \cdots \!\!&\!\! p^{(l,v)}_{2,m_l}\\
		\vdots \!\!&\!\! \vdots \!\!&\!\! \ddots \!\!&\!\! \vdots\\
		p^{(l,v)}_{m_l,1} \!\!&\!\! p^{(l,v)}_{m_l,2} \!\!&\!\! \cdots \!\!&\!\! p^{(l,v)}_{m_l,m_l}
	\end{bmatrix}\!\!
	\begin{bmatrix}
	x^{(l'\!-\!1)}_{i_1}\\
		\vspace{-0.32cm}\\
	x^{(l'\!-\!1)}_{i_2}\\
	\vdots\\
	x^{(l'\!-\!1)}_{i_{m_l}}
	\!\!\end{bmatrix}
	\!\!+\!\!
	\begin{bmatrix}
	n^{(l'\!-\!1)}_{i_1}\\
		\vspace{-0.32cm}\\
	n^{(l'\!-\!1)}_{i_2}\\
	\vdots\\
	n^{(l'\!-\!1)}_{i_{m_l}}
	\!\!\end{bmatrix}\!,\!
\end{align}
where $p^{(l,v)}_{i,j}$ is the $(i,j)$-th element of the factor matrix $\boldsymbol{P}^{(l,v)}$, $\{i_1,i_2,\dots,i_{m_l}\}$ is a length-$m_l$ subset of $\{1,2,\dots,M\}$, and $y^{(l'-1)}_i$, $x^{(l'-1)}_i$, and $n^{(l'-1)}_i$ are the $i$-th auxiliary received signal, the data symbol, and the noise component, respectively, at the beginning of the $l'$-th (end of the ($l'-1$)-st) recursion.

In order to facilitate a low-complexity recursive detection while providing the maximum increase in the effective SNRs, we obtain the effective combining coefficients {(such as combining $y_{i}$'s in \eqref{eq9})} as follows.
Given the $m_l\times m_l$ square factor matrix $\boldsymbol{P}^{(l,v)}$, as in \eqref{eq13}, we find all $T$ possible combining matrices of size $m_l\times m_l$, with the $t$-th such matrix denoted by $\boldsymbol{\alpha}^{(l,v,t)}$, for $t=1,2,\dots,T$, such that $\boldsymbol{\alpha}^{(l,v,t)}\boldsymbol{P}^{(l,v)}$ is a diagonal matrix with the nonzero diagonal entries equal to $w_i^{(l,v,t)}$, for $i=1,2,\dots,m_l$.{\footnote{If for a given $\boldsymbol{P}^{(l,v)}$ there is no any combining matrix that results in a diagonal form, with nonzero diagonal entries, for the matrix multiplication of $\boldsymbol{\alpha}^{(l,v)}\boldsymbol{P}^{(l,v)}$, we skip that matrix and do not save its attributes. Additionally, if for a given dimension $m_l$ none of $\binom{2^{m_l}-1}{m_l}$ candidates satisfy the aforementioned property, we do not consider that dimension (i.e., $m_l\times m_l$ square factor matrices) in the design of the overall pattern matrix.}} The $(i,j)$-th entry of the matrix $\boldsymbol{\alpha}^{(l,v,t)}$ is denoted by ${\alpha}^{(l,v,t)}_{i,j}\in\{-1,0,+1\}$. Now, by linearly combining the $m_l$ auxiliary received signals in \eqref{eq13} according to the rows of $\boldsymbol{\alpha}^{(l,v,t)}$ we get $m_l$ new equations as
\begin{align}\label{Eq13combined}
\sum_{j=1}^{m_l}{\alpha}^{(l,v,t)}_{i',j} y^{(l'\!-\!1)}_{i_{j}}=w_{i'}^{(l,v,t)}x^{(l'\!-\!1)}_{i_{i'}}+\sum_{j=1}^{m_l}{\alpha}^{(l,v,t)}_{i',j} n^{(l'\!-\!1)}_{i_{j}},
\end{align}
for $i'=1,2,\dots,m_l$, where $w_{i'}^{(l,v,t)}\triangleq \sum_{j=1}^{m_l}{\alpha}^{(l,v,t)}_{i',j} p^{(l,v)}_{j,i'}$.

\begin{algorithm}[t]
	\caption{Calculation of the combining coefficients and the corresponding SNR gains.}
	\begin{algorithmic}[1]
		\STATE Input: dimension of the square factor matrix $m_l$
		\STATE Output: combining matrices $\boldsymbol{\alpha}^{(l,v)}\!=\!\!\begin{bmatrix}\alpha^{(l,v)}_{i,j}\!\end{bmatrix}_{\!m_l\!\times m_l}$ and the corresponding sets $\{\gamma_{i}^{(l,v)}\}_{i=1}^{m_l}$ of SNR gains for all $v$'s
		\vspace{0.08cm}
		\FOR {$v = 1:\binom{2^{m_l}-1}{m_l}$}
		\vspace{0.08cm}
		\STATE $\boldsymbol{P}^{(l,v)}={\begin{bmatrix}p^{(l,v)}_{i,j}\end{bmatrix}}_{m_l\times m_l}$
		\vspace{0.08cm}
		\FOR {$i= 1:m_l$}
		\STATE {\textbf{Find}} all $T$ possible sets of coefficients  $\left\{{\alpha}^{(l,v,t)}_{i,j}\right\}_{\!j=1}^{\!m_l}$ {\textbf{such that:}}
		\STATE ~~~~{\textbf{C1:}} ${\alpha}^{(l,v,t)}_{i,j}\in\{-1,0,1\}$
		\vspace{0.08cm}
		\STATE ~~~~{\textbf{C2:}} $\sum_{j=1}^{m_l}{\alpha}^{(l,v,t)}_{i,j} p^{(l,v)}_{j,i}=w_{i}^{(l,v,t)}\neq 0$
		\vspace{0.12cm}
		\STATE ~~~~{\textbf{C3:}} $\sum_{j=1}^{m_l}{\alpha}^{(l,v,t)}_{i,j} p^{(l,v)}_{j,i'}=0,~i'\neq i=1,2,\dots,m_l$
		\vspace{-0.18cm}
		\STATE {\textbf{Calculate}} $\gamma_{i}^{(l,v,t)}\triangleq\left[w_{i}^{(l,v,t)}\right]^{\!2}\!\!\!{\Big/}\!\sum_{j=1}^{m_l}\left[{\alpha}^{(l,v,t)}_{i,j}\right]^{\!2}$
		\vspace{0.08cm}
		\STATE {\textbf{Calculate}} $t^*=\underset{t=1:T}{\operatorname{argmax}} ~\gamma_{i}^{(l,v,t)}$
		\vspace{0.12cm}
		\STATE {\textbf{Output}}
		$\left\{{\alpha}^{(l,v)}_{i,j}\right\}_{\!j=1}^{\!m_l}=\left\{{\alpha}^{(l,v,t^*)}_{i,j}\right\}_{\!j=1}^{\!m_l}$
		\vspace{0.08cm}
		\STATE {\textbf{Output}}
		$\gamma_{i}^{(l,v)}=\gamma_{i}^{(l,v,t^*)}$
		\ENDFOR
		\ENDFOR
	\end{algorithmic}
\end{algorithm}

The process for selecting the combining coefficients that yield the maximum SNR gain for each auxiliary data symbol is summarized in Algorithm 1. Condition \textbf{C1} in Algorithm 1 defines the set of possible values for the combining coefficients. This implies that a specific row in $\boldsymbol{P}^{(l,v)}$ is included in the combing process with a positive/negative sign or it is not included at all. Note that scaling the set of possible values $\{-1,0,+1\}$ by a constant factor does not change the SNR gains and the performance  since it scales both the data symbols and the noise coefficients in \eqref{Eq13combined} by the same multiplicative factor. Moreover, conditions \textbf{C2} and \textbf{C3} are included to make sure that the $t$-th possible combining matrix $\boldsymbol{\alpha}^{(l,v,t)}$ results in a diagonal matrix with the nonzero diagonal entries $w_i^{(l,v,t)}$'s in $\boldsymbol{\alpha}^{(l,v,t)}\boldsymbol{P}^{(l,v)}$. This is required to guarantee singleton equations, such as \eqref{Eq13combined}, in terms of the auxiliary data symbols $x^{(l'\!-\!1)}_{i_{i'}}$'s at the end of the $l'$-th recursion.
With these constraints we can ensure that after the $l'$-th recursion the maximum number of data symbols in each equation is reduced by a factor of $m_{L-l'+1}$, resulting in much simpler sets of equations. This way we come up with a very low-complexity recursive detection algorithm described in Section \ref{Sec3B2}. Note that, except for the last recursion ($l'=L$), each auxiliary data symbol $x^{(l'\!-\!1)}_{i_{i'}}$ involves a combination of several original data symbols $x_i$'s defined according to $\boldsymbol{P}^{(1)}_{m_1\times m_1}\otimes \boldsymbol{P}^{(2)}_{m_2\times m_2} \otimes \dots \otimes \boldsymbol{P}^{(L-l')}_{m_{L-l'}\times m_{L-l'}}$, i.e., the Kronecker product of the $L-l'$ leftmost factor matrices. For instance, in Example 1, the auxiliary data symbol $x^{(0)}_1$ at the beginning of the first recursion is equal to $x_1+x_5$ (see \eqref{eq9}), which is defined according to $\boldsymbol{P}^{(1)}_{3\times 3}$ as specified in \eqref{eq10_1}.

After combining the auxiliary received signals in \eqref{eq13} according to the combining coefficients $\left\{{\alpha}^{(l,v,t)}_{i',j}\right\}_{\!j=1}^{\!m_l}$, it can be observed from \eqref{Eq13combined} that the effective SNR of the $i_{i'}$-th auxiliary data symbol $x^{(l'\!-\!1)}_{i_{i'}}$ is increased by a factor of $\gamma_{i'}^{(l,v,t)}$, defined in line 10 of Algorithm 1. This is because the auxiliary noise components in \eqref{eq13} are independent, which is clarified when we describe the detection algorithm in Section \ref{Sec3B2}. Finally, Algorithm 1, among all $T$ possible sets of coefficients $\left\{{\alpha}^{(l,v,t)}_{i,j}\right\}_{\!j=1}^{\!m_l}$ that satisfy the three constraints \textbf{C1-3} for each $i$, $v$, and $l$, picks the one that maximizes the corresponding SNR gain $\gamma_{i}^{(l,v,t)}$ as the $i$-th row of the combining matrix $\boldsymbol{\alpha}^{(l,v)}$. By repeating this procedure for all $i$'s, for $i=1,2,\dots,m_l$, we get the whole combining matrix $\boldsymbol{\alpha}^{(l,v)}\!=\!\!\begin{bmatrix}\alpha^{(l,v)}_{i,j}\!\end{bmatrix}_{\!m_l\!\times m_l}$ and the corresponding sets $\{\gamma_{i}^{(l,v)}\}_{i=1}^{m_l}$ of SNR gains given the $v$-th possible square factor matrix $\boldsymbol{P}^{(l,v)}$ for the $l$-th leftmost square factor matrix used in the Kronecker product to form the overall pattern matrix. 

\noindent{\textbf{Remark 1.}} Algorithm 1 does not output a specific square factor matrix for the $l$-th position in the Kronecker product. Instead, it finds the best combining coefficients (i.e., the combining matrix $\boldsymbol{\alpha}^{(l,v)}$) and the resulting sets of SNR gains for all $\binom{2^{m_l}-1}{m_l}$ possible square factor matrices $\boldsymbol{P}^{(l,v)}_{m_l\times m_l}$. Among all these possibilities for $v$, the \textit{best} one represented by the index $v^*$ that results in the \textit{optimal} square factor matrix $\boldsymbol{P}^{(l)}$ with the corresponding set $\{\gamma_{1}^{(l)},\gamma_{2}^{(l)},\dots,\gamma_{m_l}^{(l)}\}$ of SNR gains and the optimal combining matrix $\boldsymbol{\alpha}^{(l)}$ can be obtained by choosing the best answer set satisfying further constraints such as maximizing the average sum rate (see also Remark 3). { Note that, in the case of GMAC model considered throughout Sections III and IV, Algorithm 1 needs to run only once per dimension $m_l$ of the square factor matrix (please see Remark 9 for the clarifications on the case of fading channels). Once the \textit{best} $\boldsymbol{P}^{(l)}$ and the corresponding $\boldsymbol{\alpha}^{(l)}$ are obtained for a given dimension $m_l$, they will be known to all users and the BS. Note also that this is a part of the system design and not the detection algorithm, and one needs to also take into account the number of users $K$ and available REs $M$ in the design of the overall pattern matrix and choosing the dimensions of the underlying factor matrices.}

Finally, the following lemma is useful in the process of finding the combining coefficients described in Algorithm 1.
\begin{lemma}\label{lem1}
	For a binary square factor matrix $\boldsymbol{P}$ with linearly-independent rows, we have $T\leq 1$. In other words, there exists at most one combining matrix $\boldsymbol{\alpha}=[\alpha_{i,j}]$, $\alpha_{i,j}\in\{-1,0,1\}$, that results in a diagonal form for the matrix multiplication $\boldsymbol{\alpha}\boldsymbol{P}$ with nonzero diagonal entries.
\end{lemma}
\begin{proof}
	Please refer to Appendix \ref{AppA}.
\end{proof}

\subsubsection{Recursive Detection Algorithm}\label{Sec3B2} Suppose that the communication protocol is established between the transmitters and receivers, i.e., the design parameters such as the optimal square factor matrices $\boldsymbol{P}^{(l)}_{m_l\times m_l}$, for $l=1,2,\dots,L$, and their corresponding combining matrices $\boldsymbol{\alpha}^{(l)}_{m_l\times m_l}$'s from Algorithm 1 (also recall Remark 1) are known to the receiver.
Next, we describe the low-complexity recursive detection algorithm thanks to the underlying structure of the overall pattern matrix. 
The proposed detection algorithm, {which is schematically shown through a tree diagram in Fig. \ref{tree1},} proceeds by combining the received signals according to the combining matrix of the rightmost factor matrix and involves $L$ recursions.

\textbf{First Recursion:} The receiver in the first recursion takes the vector of received signals over $M=\prod_{l=1}^{L}m_l$ REs, i.e., $y_i$'s, for $i=1,2,\dots,M$, and divides them into $M_{L-1}\triangleq M/m_L=\prod_{l=1}^{L-1}m_l$ groups of $m_L$ received signals; therefore, the $i_L$-th group, for $i_L=1,2,\dots,M_{L-1}$, 
includes $\{y_{(i_L-1)m_L+1},y_{(i_L-1)m_L+2},\dots,y_{i_Lm_L}\}$. Each of these $y_i$'s contains at most $K=\prod_{l=1}^{L}m_l$ different transmitted symbols $x_k$'s, $k=1,2,\dots,K$, since each one is constructed as a linear combination of $x_k$'s, defined with respect to the $i$-th row of the overall pattern matrix $\boldsymbol{G}_{M \times K}$, and the noise component $n_i$ over the $i$-th RE.
Note that our recursive detection algorithm attempts to reduce the maximum number of different symbols by a factor of $m_{l''}$, $l''\triangleq L-l'+1$, after each $l'$-th recursion, such that after all $L$ recursions we have $M$ equations each containing only a single data symbol.
The receiver now combines the elements of each group using the combining matrix $\boldsymbol{\alpha}^{(L)}_{m_L\times m_L}$ to form new $m_L$ symbols at each group; the first new symbol is constructed by combining the previous symbols using the first row of $\boldsymbol{\alpha}^{(L)}$ and so on. Note that the impact of such a combining on the involved data symbols can be expressed in a matrix form through multiplying $\boldsymbol{\alpha}^{(L)}$ by a new matrix comprising the $m_L$ rows of $\boldsymbol{G}$ corresponding to the indices of each group. Then it is easy to verify that the $j_L$-th new equation of each group (which is constructed through the $j_L$-th row of $\boldsymbol{\alpha}^{(L)}$), for $j_L=1,2,\dots,m_L$, contains at most $K_{L-1}\triangleq K/m_L=\prod_{l=1}^{L-1}m_l$ different symbols from the set $\{x_{j_L},x_{j_L+m_L},\dots,x_{j_L+(K_{L-1}-1)m_L}\}$ (recall that the product of $\boldsymbol{\alpha}^{(L)}\boldsymbol{P}^{(L)}$ is a diagonal matrix), i.e., the number of unknown variables is reduced by a factor of $m_L$ from $K$ to $K_{L-1}$. Also, based on the detailed discussion in Section \ref{Sec3B1} (see, e.g., Eq. \eqref{Eq13combined}), the effective SNR of each symbol in the $j_L$-th new equation is increased by a factor of $\gamma^{(L)}_{j_L}$.

\textbf{Second Recursion:} Note that after the first recursion the $j_L$-th new equation of each group can only contain data symbols from the set $\{x_{j_L},x_{j_L+m_L},\dots,x_{j_L+(K_{L-1}-1)m_L}\}$. This means different equations of a given group contain disjoint sets of data symbols while equations with the same index of different groups (e.g., the $j_L$-th equation of all groups) contain symbols from the same set. Therefore, the receiver in the second recursion forms $m_L$ super-groups of $M_{L-1}$ equations/signals by, consecutively, placing the $j_L$-th new equation of each of those $M_{L-1}$ groups of the first recursion into the $j_L$-th super-group. 
Note that the combination of the original data symbols over each super-group (in Section \ref{Sec3B1}, we denoted each of those combinations by an auxiliary data symbol obtained after the first recursion) is defined according to the Kronecker product $\boldsymbol{P}^{(1)}_{m_1\times m_1}\otimes \boldsymbol{P}^{(2)}_{m_2\times m_2} \otimes \dots \otimes \boldsymbol{P}^{(L-1)}_{m_{L-1}\times m_{L-1}}$.
Now, the receiver follows exactly the same procedure as the first recursion over each of these disjoint $m_L$ super-groups of the size $M_{L-1}$. In other words, it divides the equations in each of the super-groups into $M_{L-2}\triangleq M_{L-1}/m_{L-1}=\prod_{l=1}^{L-2}m_l$ groups of $m_{L-1}$ equations each and combines the signals within each group using the combining matrix $\boldsymbol{\alpha}^{(L-1)}_{m_{L-1}\times m_{L-1}}$. Following the same logic, we argue that the maximum number of unknown variables at each of the new equations is reduced by a factor of $m_{L-1}$ from $K_{L-1}$ to $K_{L-2}\triangleq K_{L-1}/m_{L-1}=\prod_{l=1}^{L-2}m_l$, and the SNR of the symbols in the $j_{L-1}$-st new equation, for $j_{L-1}=1,2,\dots,m_{L-1}$, of each of the groups in the $j_L$-th super-group is increased by a factor of $\gamma^{(L-1)}_{j_{L-1}}$ from $\gamma^{(L)}_{j_L}$ in the first recursion to $\gamma^{(L-1)}_{j_{L-1}}\gamma^{(L)}_{j_L}$.

\begin{figure}[t]
	\centering
	\includegraphics[trim=0.8cm 2.2cm 2.8cm 18cm ,width=3.5in]{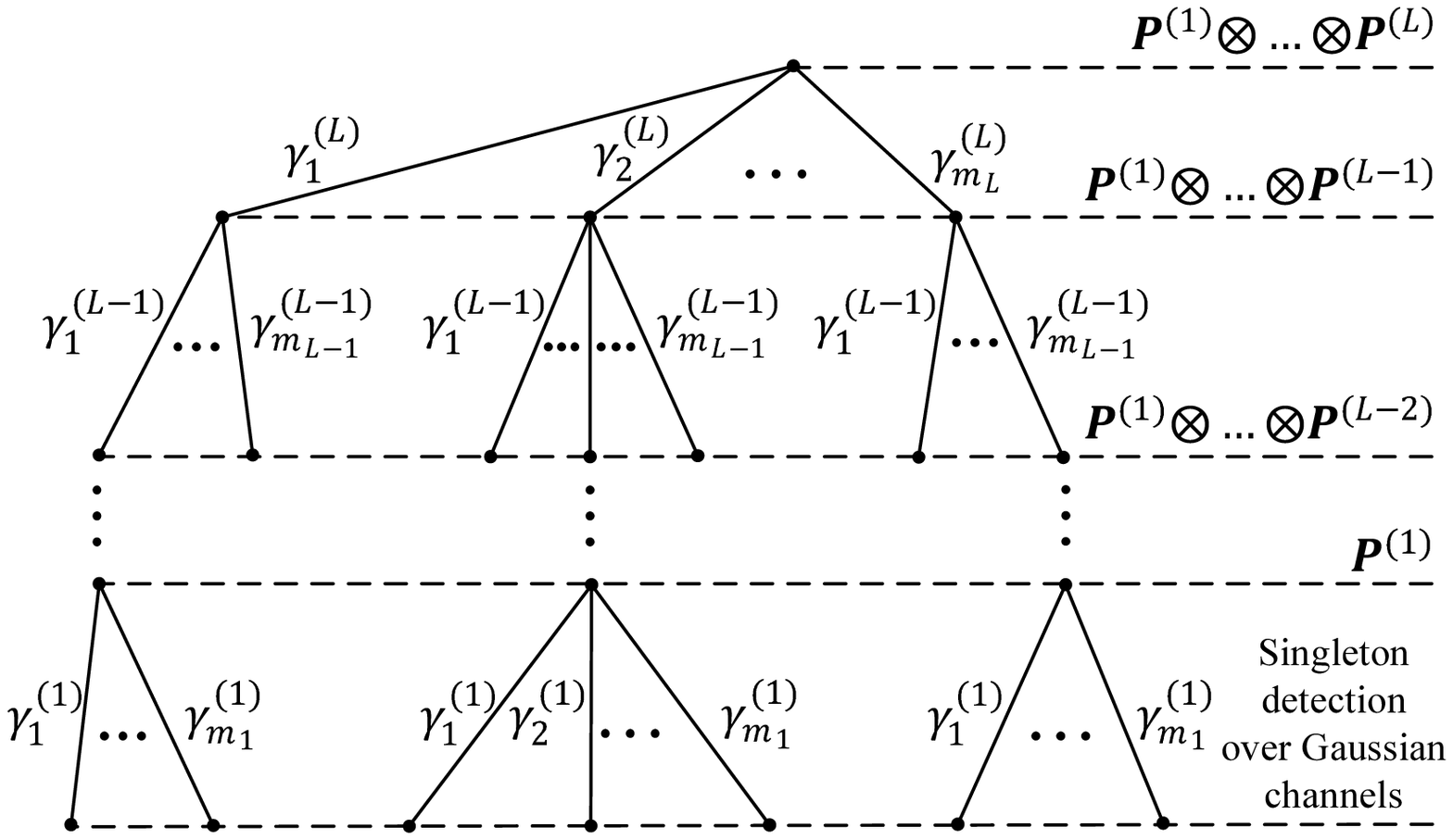}
	\caption{Tree diagram representation of the recursive detection algorithm.}
	\label{tree1}
		\vspace{-0.1in}
\end{figure}

\textbf{Final Recursion:} It can be verified by induction that in the $L$-th recursion we have $\prod_{l=2}^{L}m_l$ super-groups of size $m_1$. The sets of $m_{1}$ equations in each of these $\prod_{l=2}^{L}m_l$ super-groups is formed according to the square factor matrix $\boldsymbol{P}^{(1)}_{m_{1}\times m_{1}}$. Therefore, the receiver in the $L$-th recursion combines the $m_1$ symbols of each super-group using $\boldsymbol{\alpha}^{(1)}_{m_1\times m_1}$ to get equations containing a single unknown variable. Finally, the receiver detects the original $K=M$ data symbols according to the $M$ singleton equations in the form of a single-user additive white Gaussian noise (AWGN) channel. However, the effective SNR of each of these $K$ data symbols is increased to a desired level that can guarantee a predetermined data rate, particularly for low-capacity applications.
\subsubsection{Detection and SNR Evolution Trees}\label{Sec3B3}
For the sake of clarity, we have separately explained the second and the final recursions in Section \ref{Sec3B2}. Note that the same procedure is applied in all of the recursion steps for $l'=2,\dots,L$. In particular, the receiver at the beginning of the $l'$-th recursion, first, forms $\prod_{l=l''+1}^{L}m_l$ super-groups each containing $M_{l''}\triangleq \prod_{l=1}^{l''}m_l$ equations/signals defined according to the pattern matrix $\boldsymbol{P}^{(1)}_{m_1\times m_1}\otimes \boldsymbol{P}^{(2)}_{m_2\times m_2} \otimes \dots \otimes \boldsymbol{P}^{(l'')}_{m_{l''}\times m_{l''}}$, where $l''\triangleq L-l'+1$. The receiver then applies similar steps to the first recursion, i.e., it combines the equations/signals inside each super-group according to the rightmost combining matrix $\boldsymbol{\alpha}^{(l'')}_{m_{l''}\times m_{l''}}$ to reduce the (maximum) number of unknown data symbols included in each equation/signal by a factor of $m_{l''}$ from $K_{l''}\triangleq \prod_{l=1}^{l''}m_l$ to $K_{l''-1}$. Finally, as a result of this combining, the effective SNR of each data symbol involved in the $j_{l''}$-th combined signal, $j_{l''}=1,2,\dots,{m_{l''}}$, is increased by a factor of $\gamma^{(l'')}_{j_{l''}}$. The whole process is, schematically, shown in Fig. \ref{tree1}. Each node in the tree diagram of Fig. \ref{tree1} represents a super-group while the weights of the edges characterize the SNR gains after combining the signals inside each super-group. Moreover, the pattern matrix governing the way data symbols involved in each super-group, at the beginning and end of each recursion, are merged together is also specified in Fig. \ref{tree1} which further clarifies how after $L$ recursions we end up with $M$ singleton equations over single-user AWGN channels with increased effective SNRs.

It is worth mentioning that the $M$ overall  SNR gains obtained by multiplying the edge weights involved from the topmost node to each of the $M$ bottommost nodes in Fig. \ref{tree1} specify the set of $M$ overall SNR gains for the $M$ data symbols in an \textit{unsorted} fashion. In order to get the overall SNR gains \textit{sorted}, we need to shuffle the edge weights, i.e., the SNR gains, according to Fig. \ref{tree2}. This way we can assure that the overall SNR gain obtained by multiplying the edge weights involved from the topmost node to each $i$-th bottommost node in Fig. \ref{tree2}, for $i=1,2,\dots,M$, exactly specifies the overall SNR gain on the $i$-th data symbol $x_i$ after $L$ layers of combining. 
This can be understood from the detection algorithm elaborated in Section \ref{Sec3B2}, especially the parts emphasizing on the data symbols that remain in each new combined equation/signal and the gain those data symbols attain after each combining.
Now, according to Fig. \ref{tree2}, we have the following lemma for the overall SNR gain on the $i$-th data symbol, denoted by $\gamma_{t,i}$.
\begin{lemma}\label{lemm2_gains}
	The overall SNR gain $\gamma_{t,i}$ on the $i$-th data symbol, for $i=1,2,\dots,M$, can be obtained as
	\begin{align}\label{g_ti}
	\gamma_{t,i}=\prod_{l=1}^{L}\gamma_{s_l}^{(l)},
	\end{align}
	where the integers $s_l\in\{1,2,\dots,m_l\}$, for $l=1,2,\dots,L$, constitute a unique representation of $i$ as follows:
\begin{align}\label{s_tild}
i={s}_L+\sum_{l=1}^{L-1}({s}_{L-l}-1)\times\prod_{l_1=L-l+1}^{L}m_{l_1}.
\end{align}
Equivalently, $x_i$ is the single data symbol remained in the singleton equation of the final recursion indexed by the path $(s_L,s_{L-1},\dots,s_2,s_1)$ of super-groups  in Fig. \ref{tree1}.
\end{lemma}
\begin{figure}[t]
	\centering
	\includegraphics[trim=0.4cm 3cm 1.6cm 19cm ,width=3.6in]{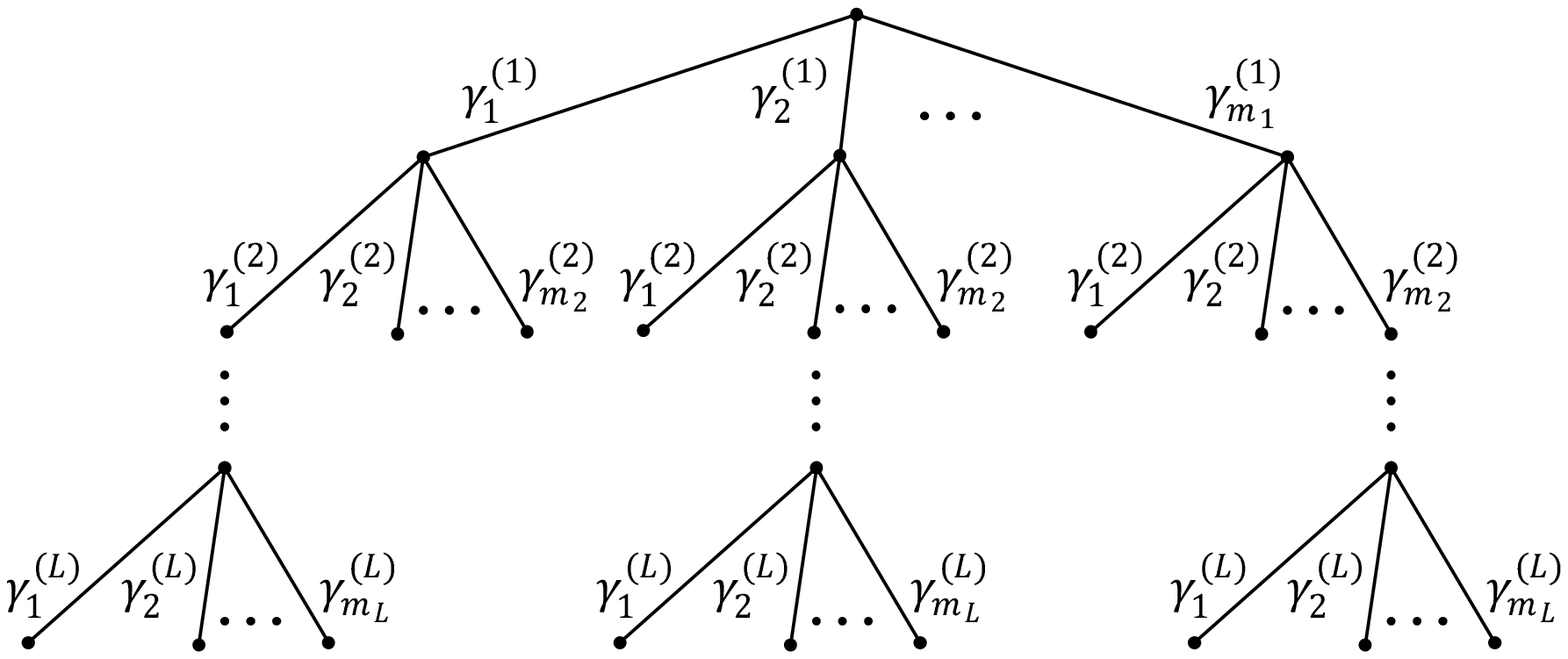}
	\caption{Tree diagram illustrating the progressive development of SNR gains.}
	\label{tree2}
		\vspace{-0.1in}
\end{figure}
\noindent{\textbf{Example 2.}} For the factor matrices considered in Example 1 and their corresponding combining matrices, according to the detailed discussions in Section \ref{Sec3B1}, we have $\gamma_{1}^{(1)}=\gamma_{2}^{(1)}=\gamma_{3}^{(1)}=\gamma_{1}^{(2)}=\gamma_{2}^{(2)}=\gamma_{3}^{(2)}=4/3$, and $\gamma_{4}^{(2)}=1$. Then, according to \Lref{lemm2_gains}, we have, for instance, $s_1=2$ and $s_2=4$ for $i=8$. Hence, $\gamma_{t,8}=\gamma_{2}^{(1)}\gamma_{4}^{(2)}=4/3$. Similarly, $\gamma_{t,4}=\gamma_{t,12}=4/3$, and $\gamma_{t,i}=(4/3)^2$, $i\notin\{4,8,12\}$.\endproofempt

\subsection{Rectangular Factor Matrices}\label{Sec3C}

In this subsection, we describe the detection algorithm when the data symbols are mixed using the Kronecker product of rectangular factor matrices over GMAC. Specifically, we consider the overall pattern matrix as $\boldsymbol{G}_{M\times K}=\boldsymbol{F}^{(1)}_{m_1\times k_1}\otimes \boldsymbol{F}^{(2)}_{m_2\times k_2} \otimes \dots \otimes \boldsymbol{F}^{(L)}_{m_L\times k_L}$, where $\boldsymbol{F}^{(l)}_{m_l\times k_l}$, for $l=1,2,\dots,L$, is an $m_l\times k_l$ binary  matrix with $k_l>m_l$. 
In this case, the overload factor $\beta= K/M$ is greater than $1$ resulting in an improved spectral efficiency compared to the case of square factor matrices considered in Section \ref{Sec3B}. We begin with the following illustrative example that helps better understand the recursive detection algorithm provided afterward. 

\noindent{\textbf{Example 3.}} Consider the transmission of $K=18$ users (symbols) over $M=4$ REs realized using the Kronecker product of the following rectangular factor matrices
\vspace{-0.05in}
\begin{align}\label{eq16}
\boldsymbol{F}^{(1)}_{1\times 2}=
\begin{bmatrix}
1 & 1
\end{bmatrix}, ~~~~ \boldsymbol{F}^{(2)}_{2\times 3}=\boldsymbol{F}^{(3)}_{2\times 3}=
\begin{bmatrix}
1 & 0 & 1 \\
1 & 1 & 0
\end{bmatrix}.
\end{align}
Hence, $\boldsymbol{y}_{4\times 1}=\boldsymbol{G}_{4\times 18}\boldsymbol{x}_{18\times 1}+\boldsymbol{n}_{4\times 1}$, with $\boldsymbol{G}_{4\times 18}=\boldsymbol{F}^{(1)}_{1\times2}\otimes\boldsymbol{F}^{(2)}_{2\times 3}\otimes\boldsymbol{F}^{(3)}_{2\times 3}$, represents the received signals vector.

The detection starts by defining $6$ auxiliary symbols $Z^{(3)}_{3i+j}$, for $i=0,1$ and $j=1,2,3$, each obtained by spreading the $j$-th data symbol of $\boldsymbol{F}^{(3)}$ according to the $(i+1)$-st row of $\boldsymbol{R}^{(2)}\triangleq \boldsymbol{F}^{(1)}\otimes\boldsymbol{F}^{(2)}$, i.e., 
\begin{align}\label{eq17}
Z^{(3)}_{j}=x_{j}+x_{j+6}+x_{j+9}+x_{j+15},\nonumber\\
Z^{(3)}_{3+j}=x_{j}+x_{j+3}+x_{j+9}+x_{j+12}.
\end{align}
Now, it is easy to observe that the original set of equations can be rewritten as the following two sets of equations each defined according to $\boldsymbol{F}^{(3)}$ as
\begin{align}\label{eq18}
\begin{bmatrix}
y_1 \\
y_2
\end{bmatrix}=\boldsymbol{F}^{(3)}_{2\times 3}\begin{bmatrix}
Z^{(3)}_{1}&
Z^{(3)}_{2}&
Z^{(3)}_{3}
\end{bmatrix}^T+\begin{bmatrix}
n_1 \\
n_2
\end{bmatrix},\nonumber\\
\begin{bmatrix}
y_3 \\
y_4
\end{bmatrix}=\boldsymbol{F}^{(3)}_{2\times 3}\begin{bmatrix}
Z^{(3)}_{4}&
Z^{(3)}_{5}&
Z^{(3)}_{6}
\end{bmatrix}^T+\begin{bmatrix}
n_3 \\
n_4
\end{bmatrix}.
\end{align}
Therefore, at the beginning of the first recursion in the receiver, the sets of equations are formed according to $\boldsymbol{F}^{(3)}_{2\times 3}$ as \eqref{eq18}.
{The receiver first detects $Z^{(3)}_{1},Z^{(3)}_{2},\dots,Z^{(3)}_{6}$, using any MUD algorithm performed over the set of equations in \eqref{eq18}, and then groups $Z^{(3)}_{j}$ with $Z^{(3)}_{3+j}$ to obtain three sets of equations each defined according to $\boldsymbol{R}^{(2)}$.}
Let us consider the first group containing $Z^{(3)}_{1}$ and $Z^{(3)}_{4}$. According to \eqref{eq17}, 
\begin{align}\label{eq19}
\begin{bmatrix}
Z^{(3)}_{1} \\
Z^{(3)}_{4}
\end{bmatrix}=\underbrace{\begin{bmatrix}
	1 & 0 & 1 \\
	1 & 1 & 0
	\end{bmatrix}}_{{\boldsymbol{F}^{(2)}_{2\times 3}}}\begin{bmatrix}
Z^{(2)}_{1}\\
Z^{(2)}_{2}\\
Z^{(2)}_{3}
\end{bmatrix},
\end{align} 
after setting $Z^{(2)}_{1}\triangleq x_{1}+x_{10}$, $Z^{(2)}_{2}\triangleq x_{4}+x_{13}$, $Z^{(2)}_{3}\triangleq x_{7}+x_{16}$ each according to $\boldsymbol{R}^{(1)}\triangleq \boldsymbol{F}^{(1)}$. Now, the receiver first detects $Z^{(2)}_{1}$, $Z^{(2)}_{2}$, and $Z^{(2)}_{3}$, {using any arbitrary MUD algorithm,} from the system of equations in \eqref{eq19} defined according to ${\boldsymbol{F}^{(2)}_{2\times 3}}$, and then detects the corresponding $6$ original data symbols from the three subsequent sets of equations each defined according to $\boldsymbol{F}^{(1)}$. Finally, by applying the same procedure to the other two groups, we detect all of the $18$ original data symbols. \endproofempt

The above example illustrates how the detection proceeds recursively by processing, at each recursion, sets of equations with smaller number of variables each formed according to one of the rectangular factor matrices.
In the following, we present the detailed description of the recursive detection algorithm over rectangular factor matrices. Same as in Section \ref{Sec3B2}, we assume that the rectangular factor matrices $\boldsymbol{F}^{(l)}_{m_l\times k_l}$, for $l=1,2,\dots,L$,
are known to the receiver. 
The proposed detection algorithm entails $L$ recursions while the $l'$-th recursion, for $l'=1,2,\dots,L$, involves sets of equations formed by the factor matrix $\boldsymbol{F}^{(l'')}_{m_{l''}\times k_{l''}}$ with $l''\triangleq L-l'+1$. Note that when $L=1$ we directly detect the $K$ data symbols mixed over $M$ REs using the single factor matrix. Hence, in the algorithm it is assumed that $L>1$.
{Also, no specific structure is imposed here on the rectangular factor matrices. Instead, we assume a generic form for the factor matrices such that a given advanced MUD algorithm can properly work to detect the unknown variables  of the $l'$-th recursion defined according to $\boldsymbol{F}^{(l'')}_{m_{l''}\times k_{l''}}$. However, one can impose further constraints on the rectangular factor matrices to improve the detection performance of the scheme at each recursion. For example, one can potentially apply the results of \cite{song2020super} to more efficiently design each of the middle rectangular factor matrices and improve the detection performance of the MUDs performed over each of them at each recursion. Further investigation on this matter is left for future studies.}
Throughout this subsection, we define $\boldsymbol{R}^{(l)}_{M_l\times K_l}\triangleq \boldsymbol{F}^{(1)}_{m_1\times k_1}\otimes \boldsymbol{F}^{(2)}_{m_2\times k_2}\otimes \dots \otimes\boldsymbol{F}^{(l)}_{m_l\times k_l}$ with $M_l\triangleq\prod_{i=1}^{l}m_i$ and $K_l\triangleq\prod_{i=1}^{l}k_i$, $l=1,2,\dots,L$.

\textbf{First Recursion:} Given that $\boldsymbol{G}_{M\times K}=\boldsymbol{R}^{(L-1)}_{M_{L-1}\times K_{L-1}}\otimes\boldsymbol{F}^{(L)}_{m_L\times k_L}$, in the first recursion, the receiver forms sets of equations each defined according to $\boldsymbol{F}^{(L)}_{m_L\times k_L}$. To this end, it detects $Q_L\triangleq k_LM_{L-1}$ new/auxiliary symbols $Z^{(L)}_{k_Li_L+j_L}$, for $i_L=0,1,\dots,M_{L-1}-1$ and $j_L=1,2,\dots,k_L$, each defined as the expansion of the $j_L$-th data symbol of $\boldsymbol{F}^{(L)}_{m_L\times k_L}$ according to the $(i_L+1)$-st row of
$\boldsymbol{R}^{(L-1)}$ as
\begin{align}\label{eq20}
\!\!Z^{(L)}_{k_Li_L+j_L}\!\!\!=\!\boldsymbol{r}^{(L-1)}_{i_L+1}\!\begin{bmatrix}
x_{j_L}\!\!&\!\!x_{j_L+k_L}\!\!&\!\!\dots\!\!&\!\!x_{j_L+(K_{L-1}-1)k_L}
\end{bmatrix}\!{^T}\!\!,\!
\end{align}
where $\boldsymbol{r}^{(L-1)}_{i_L+1}$ is the $(i_L+1)$-st row of $\boldsymbol{R}^{(L-1)}_{M_{L-1}\times K_{L-1}}$. Using \eqref{eq20}, the original set of equations $\boldsymbol{y}_{M\times 1}=\boldsymbol{G}_{M\times K}\boldsymbol{x}_{K\times 1}+\boldsymbol{n}_{M\times 1}$ can be decomposed into $M_{L-1}$ sets of equations each having $m_L$ consecutive $y_i$'s as the input and $k_L$ consecutive auxiliary symbols $Z^{(L)}_{k_Li_L+j_L}$'s as the unknown variables. The $(i_L+1)$-st such set of equations is expressed as
\begin{align}\label{eq21}
\!\begin{bmatrix}
y_{i_Lm_L+1} \\
y_{i_Lm_L+2}\\
\vdots\\
y_{(i_L+1)m_L}
\end{bmatrix}\!=\boldsymbol{F}^{(L)}_{m_L\times k_L}\!\begin{bmatrix}
Z^{(L)}_{i_Lk_L+1}\\
\vspace{-0.3cm}\\
Z^{(L)}_{i_Lk_L+2}\\
\vspace{-0.5cm}\\
\vdots\\
\vspace{-0.5cm}\\
Z^{(L)}_{(i_L+1)k_L}
\end{bmatrix}\!+\!\begin{bmatrix}
n_{i_Lm_L+1} \\
n_{i_Lm_L+2}\\
\vdots\\
n_{(i_L+1)m_L}
\end{bmatrix}\!.\!
\end{align}
Therefore, the first recursion involves $M_{L-1}$ separate sets of equations, each defined according to the rectangular factor matrix $\boldsymbol{F}^{(L)}_{m_L\times k_L}$, as specified in \eqref{eq21}, such that $Q_L$ auxiliary symbols are detected. Next, the receiver forms $k_L$ disjoint super-groups, given that all of the $M_{L-1}$ auxiliary symbols $Z^{(L)}_{k_Li_L+j_L}$, for $i_L=0,1,\dots,M_{L-1}-1$, contain the same disjoint set of data symbols for a fixed $j_L$ (see \eqref{eq20}). The $j_L$-th super-group contains $M_{L-1}$ auxiliary symbols $Z^{(L)}_{j_L},Z^{(L)}_{k_L+j_L},\dots,Z^{(L)}_{(M_{L-1}-1)k_L+j_L}$. Therefore, the set of equations in the $j_L$-th super-group can be expressed as
\begin{align}\label{eq22}
\!\!\begin{bmatrix}
Z^{(L)}_{j_L}\\
\vspace{-0.3cm}\\
Z^{(L)}_{k_L+j_L}\\
\vspace{-0.5cm}\\
\vdots\\
\vspace{-0.65cm}\\
Z^{(L)}_{(M_{L-1}\!-\!1)k_L\!+j_L}
\end{bmatrix}\!\!=\!\boldsymbol{R}^{(L-1)}_{M_{L-1}\times K_{L-1}}\!\!\begin{bmatrix}
x_{j_L} \\
x_{k_L+j_L}\\
\vdots\\
\vspace{-0.5cm}\\
x_{(K_{L-1}\!-\!1)k_L\!+j_L}
\end{bmatrix}\!\!.\!
\end{align}

\textbf{Middle Recursions ($1<l'<L$):} When $L>2$, the detection algorithm involves $L-2$ middle recursions excluding the first and the last recursions. Each of these middle recursions involves a procedure that is somewhat similar to the first recursion. For example, in the second recursion we start with $k_L$ disjoint super-groups each containing a set of equations as \eqref{eq22} and then we apply similar steps as the first recursion to this set of equations. However, since the indices of the data and the auxiliary symbols matter in the detection done at each recursion, here, we carefully describe these middle recursions. In general, at the $l'$-th recursion, for $1<l'<L$, the algorithm processes the $\kappa^{\mathcal{S}}_{l'-1}\triangleq\prod_{i=L-l'+2}^{L}k_i$ disjoint super-groups of the $(l'-1)$-st recursion separately. It then comes up with $k_{L-l'+1}$ smaller size super-groups for each of those $\kappa^{\mathcal{S}}_{l'-1}$ starting super-groups, resulting in $k_{L-l'+1}\kappa^{\mathcal{S}}_{l'-1}=\kappa^{\mathcal{S}}_{l'}$ super-groups at the end. Particularly, consider the path $(j_L,j_{L-1},\dots,j_{L-l'+2})$ of the super-groups from the previous recursions, with $j_l=1,2,\dots,k_l$ for $l=1,2,\dots,L$, (a similar tree diagram to Fig. \ref{tree1} can be considered for the recursive detection here). The indices in the path correspond to the super-groups considered from the previous recursions while there are total of $\kappa^{\mathcal{S}}_{l'-1}$ such paths, i.e., starting super-groups. By induction, the set of equations in the super-group indexed by the path $(j_L,j_{L-1},\dots,j_{L-l'+2})$ can be expressed as follows:
\begin{align}\label{eq23}
\begin{bmatrix}
Z^{(l''+1)}_{\psi_{l'-1},1}\\
\vspace{-0.3cm}\\
Z^{(l''+1)}_{\psi_{l'-1},2}\\
\vspace{-0.57cm}\\
\vdots\\
\vspace{-0.5cm}\\
Z^{(l''+1)}_{\psi_{l'-1},M_{l''}}
\end{bmatrix}=\boldsymbol{R}^{(l'')}_{M_{l''}\times K_{l''}}\begin{bmatrix}
x_{\tau_{l'-1}} \\
x_{\kappa^{\mathcal{S}}_{l'-1}+\tau_{l'-1}}\\
\vspace{-0.45cm}\\
\vdots\\
\vspace{-0.5cm}\\
x_{(K_{l''}-1)\kappa^{\mathcal{S}}_{l'-1}+\tau_{l'-1}}
\end{bmatrix}\!,\!
\end{align}
where $\psi_{l'-1}\triangleq j_{L-l'+2}+\sum_{i=1}^{l'-2}(j_{L-i+1}-1)\prod_{i'=L-l'+2}^{L-i}k_{i'}$ is the index of the super-group represented by the path $(j_L,j_{L-1},\dots,j_{L-l'+2})$, and $\tau_{l'-1}\triangleq j_L+\sum_{i=1}^{l'-2}(j_{L-i}-1)\kappa^{\mathcal{S}}_{i}$. Moreover, $Z^{(l''+1)}_{\psi_{l'-1},1},\dots,Z^{(l''+1)}_{\psi_{l'-1},M_{l''}}$ are $M_{l''}$ auxiliary symbols defined in the $(l'-1)$-st recursion, similar to \eqref{eq20}, corresponding to the $\psi_{l'-1}$-st super-group.
Note that tracking the indices of the auxiliary symbols is not needed since the algorithm processes each super-group separately. More specifically, these symbols are defined at the beginning of a recursion and they are known when the detection is done at the end of that recursion. However, the indices of the data symbols matters since we need to know the data symbols in each super-group/path and the order they appeared in the corresponding equations. 

The algorithm involves $\kappa^{\mathcal{S}}_{l'-1}$ sets of equations at the beginning of the $l'$-th recursion. Each such a set of equations is processed separately and simpler sets of equations defined according to $\boldsymbol{F}^{(l'')}_{m_{l''}\times k_{l''}}$ are derived. Let us consider the $\psi_{l'-1}$-st super-group with the set of equations given by \eqref{eq23}. 
 Similar to the first recursion, the receiver assumes
 $Q_{l''}\triangleq k_{l''}M_{L-l'}$ auxiliary symbols $Z^{(l'')}_{k_{l''}i_{l''}+j_{l''}}$, for $i_{l''}=0,1,\dots,M_{{L-l'}}-1$ and $j_{l''}=1,2,\dots,k_{l''}$, each defined as the expansion of the $j_{l''}$-th data symbol according to $\boldsymbol{r}^{(L-l')}_{i_{l''}+1}$, i.e., the $(i_{l''}+1)$-st row of $\boldsymbol{R}^{(L-l')}_{M_{L-l'}\times K_{L-l'}}$, as
\begin{align}\label{eq24}
\!\!\!Z^{({l''})}_{k_{l''}i_{l''}+j_{l''}}\!\!\!=\!\boldsymbol{r}^{(L-l')}_{i_{l''}+1}\!\begin{bmatrix}
x_{\tau_{l'}}\!\!&\!\!\!x_{\tau_{l'}+\kappa^{\mathcal{S}}_{l'}}\!\!&\!\!\!\!\dots\!\!\!\!&\!\!x_{\tau_{l'}+(K_{L-l'}\!-1)\kappa^{\mathcal{S}}_{l'}}\!
\end{bmatrix}\!\!{^T}\!\!.\!
\end{align}
Note that the $j_{l''}$-th data symbol of \eqref{eq23} is $x_{\tau_{l'}}$ since $\tau_{l'-1}+(j_{l''}-1)\kappa^{\mathcal{S}}_{l'-1}\triangleq\tau_{l'}$. Now, instead of detecting the original data symbols from rather complex equations such as \eqref{eq23}, the receiver forms $M_{L-l'}$ simpler sets of equations each defined according to the rectangular factor matrix $\boldsymbol{F}^{(l'')}_{m_{l''}\times k_{l''}}$. The $(i_{l''}+1)$-st such set of equations can be expressed as
\begin{align}\label{eq25}
\!\begin{bmatrix}
Z^{(l''+1)}_{\psi_{l'-1},i_{l''}m_{l''}+1}\\
\vspace{-0.3cm}\\
Z^{(l''+1)}_{\psi_{l'-1},i_{l''}m_{l''}+2}\\
\vspace{-0.5cm}\\
\vdots\\
\vspace{-0.6cm}\\
Z^{(l''+1)}_{\psi_{l'-1},(i_{l''}+1)m_{l''}}
\end{bmatrix}=\boldsymbol{F}^{(l'')}_{m_{l''}\times k_{l''}}\begin{bmatrix}
Z^{(l'')}_{i_{l''}k_{l''}+1}\\
\vspace{-0.3cm}\\
Z^{(l'')}_{i_{l''}k_{l''}+2}\\
\vspace{-0.5cm}\\
\vdots\\
\vspace{-0.5cm}\\
Z^{(l'')}_{(i_{l''}+1)k_{l''}}
\end{bmatrix}.
\end{align}
After detecting the $Q_{l''}$ auxiliary symbols $Z^{(l'')}_{k_{l''}i_{l''}+j_{l''}}$'s from the sets of equations as \eqref{eq25}, the receiver then forms $k_{l''}$ disjoint super-groups, given that all of the $M_{L-l'}$ auxiliary symbols $Z^{(l'')}_{k_{l''}i_{l''}+j_{l''}}$, $i_{l''}=0,1,\dots,M_{L-l'}-1$, contain the same disjoint set of data symbols for a fixed $j_{l''}$. The set of equations in the $j_{l''}$-th super-group can be expressed as
\begin{align}\label{eq26}
\begin{bmatrix}
Z^{(l'')}_{\psi_{l'},1}\\
\vspace{-0.3cm}\\
Z^{(l'')}_{\psi_{l'},2}\\
\vspace{-0.55cm}\\
\vdots\\
\vspace{-0.55cm}\\
Z^{(l'')}_{\psi_{l'},M_{l''-1}}
\end{bmatrix}\!=\!\boldsymbol{R}^{(l''-1)}_{M_{l''-1}\times K_{l''-1}}\begin{bmatrix}
x_{\tau_{l'}} \\
x_{\kappa^{\mathcal{S}}_{l'}+\tau_{l'}}\\
\vdots\\
x_{(K_{l''-1}-1)\kappa^{\mathcal{S}}_{l'}+\tau_{l'}}
\end{bmatrix},
\end{align}
given that $l''-1=L-l'$. Note that in \eqref{eq26}  $Z^{(l'')}_{k_{l''}i_{l''}+j_{l''}}$ is denoted by $Z^{(l'')}_{\psi_{l'},i_{l''}+1}$ in order to be consistent with \eqref{eq23} given that the index of the super-group represented by the path $(j_L,j_{L-1},\dots,j_{L-l'+2},j_{L-l'+1})$ is $\psi_{l'}\triangleq j_{L-l'+1}+\sum_{i=1}^{l'-1}(j_{L-i+1}-1)\prod_{i'=L-l'+1}^{L-i}k_{i'}$. Finally, repeating the same procedure for all of the $\kappa^{\mathcal{S}}_{l'-1}$ starting super-groups results in $k_{L-l'+1}\kappa^{\mathcal{S}}_{l'-1}=\kappa^{\mathcal{S}}_{l'}$ super-groups, each of the form given by \eqref{eq26}, at the end of the $l'$-th recursion.

\textbf{Final Recursion:} The receiver in the $L$-th recursion starts with $\kappa^{\mathcal{S}}_{L-1}\triangleq\prod_{i=2}^{L}k_i$ super-groups each containing a disjoint subset of the original data symbols of size $k_1$. The set of equations in the $\psi_{L-1}$-st such super-group, where $\psi_{L-1}\triangleq j_{2}+\sum_{i=1}^{L-2}(j_{L-i+1}-1)\prod_{i'=2}^{L-i}k_{i'}$, indexed by the path $(j_L,j_{L-1},\dots,j_{2})$ can be expressed as (see, e.g., \eqref{eq23}):
\begin{align}\label{eq27}
\begin{bmatrix}
Z^{(2)}_{\psi_{L-1},1}\\
\vspace{-0.35cm}\\
Z^{(2)}_{\psi_{L-1},2}\\
\vspace{-0.55cm}\\
\vdots\\
\vspace{-0.55cm}\\
Z^{(2)}_{\psi_{L-1},m_1}
\end{bmatrix}=\boldsymbol{F}^{(1)}_{m_1\times k_1}\begin{bmatrix}
x_{\tau_{L-1}} \\
x_{\kappa^{\mathcal{S}}_{L-1}+\tau_{L-1}}\\
\vspace{-0.45cm}\\
\vdots\\
\vspace{-0.45cm}\\
x_{(k_1-1)\kappa^{\mathcal{S}}_{L-1}+\tau_{L-1}}
\end{bmatrix},
\end{align}
where $\tau_{L-1}\triangleq j_L+\sum_{i=1}^{L-2}(j_{L-i}-1)\kappa^{\mathcal{S}}_{i}$, and $Z^{(2)}_{\psi_{L-1},1},\dots,Z^{(2)}_{\psi_{L-1},m_1}$ are $m_1$ auxiliary symbols defined in the $(L-1)$-st recursion and their values are known at the beginning of the $L$-th recursion. Finally, after detecting the $k_1$ disjoint data symbols of each of the $\kappa^{\mathcal{S}}_{L-1}$ super-groups of the form \eqref{eq27}, the receiver obtains the values of all $K$ data symbols. Note that the indices of the data symbols contained in each super-group is fully known according to \eqref{eq27}.

\noindent{\textbf{Remark 2.}} The recursive algorithm proposed in this subsection can readily be applied to the Kronecker product of any general square factor matrices by inserting $k_l=m_l$, $l=1,2,..,L$. In that case, any given MUD scheme can be applied to sets of equations defined according to the corresponding factor matrix at each recursion. This is while the algorithm presented in Section \ref{Sec3B2} provides a very low-complexity detection.

\subsection{General Pattern Matrices}\label{Sec3D}
In this subsection, we turn our attention to the general case of pattern matrices by building upon the analysis in Sections \ref{Sec3B} and \ref{Sec3C} specifically on square and rectangular factor matrices, respectively. In particular, the general case of the pattern matrix is considered as
\begin{align}\label{eq28}
\boldsymbol{G}_{M\times K}=&\boldsymbol{F}^{(1)}_{m_1\times k_1}\otimes\boldsymbol{F}^{(2)}_{m_2\times k_2}\otimes \dots \otimes \boldsymbol{F}^{({L_r})}_{m_{L_r}\times k_{L_r}}\nonumber\\
&\!\otimes\boldsymbol{P}^{(1)}_{m'_1\times m'_1}\!\otimes\boldsymbol{P}^{(2)}_{m'_2\times m'_2}\otimes \dots \otimes \boldsymbol{P}^{({L_s})}_{m'_{L_s}\times m'_{L_s}},
\end{align}
where $L_r$ is the number of rectangular factor matrices $\boldsymbol{F}^{({l_r})}_{m_{l_r}\times k_{l_r}}$'s, for ${l_r}=1,2,\dots,L_r$, and 
$L_s$ is the number of square factor matrices $\boldsymbol{P}^{({l_s})}_{m'_{l_s}\times m'_{l_r}}$'s, for ${l_s}=1,2,\dots,L_s$, in the scheme. Moreover, let $M_r\triangleq \prod_{l_r=1}^{L_r}m_{l_r}$, $K_r\triangleq \prod_{l_r=1}^{L_r}k_{l_r}$, and $M_s\triangleq \prod_{l_s=1}^{L_s}m'_{l_s}$. Then $M=M_rM_s$, $K=K_rM_s$, and the overall overload factor $\beta$ is $K_r/M_r>1$. This is because $k_{l_r}>m_{l_r}$ for ${l_r}=1,2,\dots,L_r$. In addition to the structure of the overall pattern matrix given by \eqref{eq28}, it is assumed that the combining matrices $\boldsymbol{\alpha}^{(l_s)}_{m'_{l_s}\times m'_{l_s}}$'s corresponding to each $\boldsymbol{P}^{({l_s})}_{m'_{l_s}\times m'_{l_r}}$ are also known to the receiver.
As discussed in Section \ref{Sec3B2}, our proposed detection algorithm with only the Kronecker product of $L_s$ square factor matrices as the pattern matrix results in $M_s$ singleton equations with the SNR gain on each equation obtained using \Lref{lemm2_gains}.\footnote{{Note that the proposed detection algorithm for the Kronecker product of square factor matrices in Section \ref{Sec3B} is different than the proposed detection algorithm in Section \ref{Sec3C} for the Kronecker product of rectangular factor matrices.
		Indeed, the former cannot be obtained as a special case of the latter simply by inserting $k_{l}=m_{l}$. In particular, we carefully designed the square factor matrices and the corresponding combining matrices to enable a much more efficient detection algorithm for the Kronecker product of square factor matrices.}} The detection algorithm for the general case of the pattern matrix is then summarized as the following two phases each incurring $L_s$ and $L_r$ recursions, respectively.

\textbf{Phase I:} In the first phase, the receiver applies $L_s$ recursions, by following combining and (super-) grouping as described in Section \ref{Sec3B2}, to the system of equation $\boldsymbol{y}_{M\times 1}=\boldsymbol{G}_{M\times K}\boldsymbol{x}_{K\times 1}+\boldsymbol{n}_{M\times 1}$ with $\boldsymbol{G}_{M\times K}$ defined in \eqref{eq28}. Then the receiver obtains $M_s$ super-groups each containing a set of equations defined according to the pattern matrix $\boldsymbol{F}_{M_r\times K_r}\triangleq\boldsymbol{F}^{(1)}_{m_1\times k_1}\otimes\boldsymbol{F}^{(2)}_{m_2\times k_2}\otimes \dots \otimes \boldsymbol{F}^{({L_r})}_{m_{L_r}\times k_{L_r}}$.
The set of equations in the $\psi'_{L_s}$-th super-group, where $\psi'_{L_s}\triangleq j_{1}+\sum_{i=1}^{L_s-1}(j_{L_s-i+1}-1)\prod_{i'=1}^{L_s-i}m'_{i'}$, indexed by the path of super-groups $(j_{L_s},j_{L_s-1},\dots,j_2,j_1)$, for $j_{l_s}=1,2,\dots,m'_{l_s}$, is expressed as
\begin{align}\label{eq29}
\!\!\begin{bmatrix}
\tilde{y}_{\psi'_{L_s},1} \\
\tilde{y}_{\psi'_{L_s},2}\\
\vspace{-0.55cm}\\
\vdots\\
\vspace{-0.65cm}\\
\tilde{y}_{\psi'_{L_s},M_r}
\end{bmatrix}\!\!=\!\boldsymbol{F}_{M_r\times K_r}\!\begin{bmatrix}
x_{\tau'_{L_s}} \\
x_{M_s+\tau'_{L_s}}\\
\vspace{-0.55cm}\\
\vdots\\
\vspace{-0.65cm}\\
x_{(K_{r}-1)M_s+\tau'_{L_s}}
\end{bmatrix}\!\!+\!\!\begin{bmatrix}
\tilde{n}_{\psi'_{L_s},1} \\
\tilde{n}_{\psi'_{L_s},2}\\
\vspace{-0.55cm}\\
\vdots\\
\vspace{-0.65cm}\\
\tilde{n}_{\psi'_{L_s},M_r}
\end{bmatrix}\!,\!
\end{align}
where $\tau'_{L_s}\triangleq j_{L_s}+\sum_{i=1}^{L_s-1}(j_{L_s-i}-1)\prod_{i'=L_s-i+1}^{L_s}m'_{i'}$. Moreover, $\tilde{y}_{\psi'_{L_s},1}, \tilde{y}_{\psi'_{L_s},2},\dots,\tilde{y}_{\psi'_{L_s},M_r}$ are the combined versions of the channel outputs in $\boldsymbol{y}_{M\times 1}$, after $L_s$ rounds of combining, and their values are fully known to the receiver at the end of Phase I. Furthermore, $\tilde{n}_{\psi'_{L_s},1}, \tilde{n}_{\psi'_{L_s},2},\dots,\tilde{n}_{\psi'_{L_s},M_r}$ are independent noise terms each having a variance equal to $\sigma^2/\gamma_{t,\tau'_{L_s}}$. Note that, as a result of $L_s$ square factor matrices, the effective SNR of the equations in the $\psi'_{L_s}$-th super-group in \eqref{eq29} is increased by a factor of $\gamma_{t,\tau'_{L_s}}$, where, by \Lref{lemm2_gains},\footnote{Observe that $j_{l_s}$ here has the same meaning as $s_{l_s}$ defined in \Lref{lemm2_gains} and $\tau'_{L_s}$ has the same value as the variable $i$ defined in Eq. \eqref{s_tild}.} $\gamma_{t,\tau'_{L_s}}\triangleq \prod_{l_s=1}^{L_s}\gamma_{j_{l_s}}^{(l_s)}$ with $\gamma_{j_{l_s}}^{(l_s)}$ defined the same way as in Section \ref{Sec3B} for $\boldsymbol{P}^{({l_s})}_{m'_{l_s}\times m'_{l_r}}$.

\textbf{Phase II:} The receiver in the second phase separately processes each of the $M_s$ disjoint sets of equations of the form given by \eqref{eq29}. According to the detailed description of the detection algorithm in Section \ref{Sec3C}, the receiver, after applying $L_r$ recursions, detects the $K_r$ symbols of each of the aforementioned $M_s$ sets of equations resulting in $K=K_rM_s$ data symbols in total.

\section{Performance Characterization}\label{Sec4}
\subsection{Average Sum-Rate}\label{Sec4A}
Based on the detailed analysis provided in Section \ref{Sec3}, we have the following theorem on the average sum-rate of the proposed scheme over GMAC{\footnote{Note that, in general, the average sum-rate should depend on the statistical properties of the channel. One needs to account on our explanations in Section \ref{Sec5A} to extend \Tref{sum_rate_thm} to the case of fading channels.}}.
\begin{theorem}\label{sum_rate_thm}
The per-RE average sum-rate $C_M$ of the proposed code-domain NOMA with the general pattern matrix $\boldsymbol{G}_{M \times K}$, specified in \eqref{eq28}, and the proposed recursive detection algorithm is expressed as
\begin{align}\label{eq30}
C_M=\frac{1}{2M}\sum_{j_{L_s}=1}^{m'_{L_s}}&~\!\sum_{j_{L_s-1}=1}^{m'_{L_s-1}}\dots\sum_{j_{2}=1}^{m'_{2}}\sum_{j_{1}=1}^{m'_{1}}
\nonumber\\
&\log_2\det\!\left({\boldsymbol{I}}_{M_r}+\rho\prod_{l_s=1}^{L_s}\gamma_{j_{l_s}}^{(l_s)}\boldsymbol{F}\boldsymbol{F}^T\right),
\end{align}
where ${\boldsymbol{I}}_{M_r}$ is an $M_r\times M_r$ identity matrix, $\boldsymbol{F}_{M_r\times K_r}\triangleq\boldsymbol{F}^{(1)}_{m_1\times k_1}\otimes\boldsymbol{F}^{(2)}_{m_2\times k_2}\otimes \dots \otimes \boldsymbol{F}^{({L_r})}_{m_{L_r}\times k_{L_r}}$, $\rho\triangleq P_x/\sigma^2$ with $P_x\triangleq\E\left[x^2_j\right]$, for $j=1,2,\dots,K$, and $\sigma^2=\E\left[n^2_i\right]$, for $i=1,2,\dots,M$.
\end{theorem}
\begin{proof}
It is well known that for a PDMA system with the pattern matrix $\boldsymbol{A}_{M\times K}$ and the SNR of $\rho$ for all of the received original data symbols, the per-RE average sum rate is given by $C^{\rm PDMA}_M=\frac{1}{2M}\log_2\det\left({\boldsymbol{I}}_{M}+\rho\boldsymbol{A}\boldsymbol{A}^T\right)$ for the optimal MAP detection \cite{shental2017low}. As elaborated in Phase I of the detection algorithm in Section \ref{Sec3D}, we end up with $M_s$ sets of equation each defined according to \eqref{eq29}. Therefore, the per-RE average sum-rate $C_{M_r}(\psi'_{L_s})$ of the $\psi'_{L_s}$-th super-group, indexed by the path $(j_{L_s},j_{L_s-1},\dots,j_2,j_1)$, is obtained as
\begin{align}\label{eq31}
C_{M_r}(\psi'_{L_s})=\frac{1}{2M_r}\log_2\det\left({\boldsymbol{I}}_{M_r}+\rho\gamma_{t,\tau'_{L_s}}\boldsymbol{F}\boldsymbol{F}^T\right).
\end{align}
Averaging \eqref{eq31} over all $M_s$ paths completes the proof.
\end{proof}

In the special case with all the square factor matrices being identical, i.e., $\boldsymbol{P}^{({l_s})}_{m'_{l_s}\times m'_{l_r}}=\boldsymbol{P}_{m_p\times m_p}$, for ${l_s}=1,2,\dots,L_s$, we have the following corollary given that $\gamma_{j}^{(l_s)}=\gamma_{j}$, for $j=1,2,\dots,m_p$.
\begin{corollary}\label{sum_rate_cor}
	The per-RE average sum-rate of the proposed code-domain NOMA with the pattern matrix $\boldsymbol{G}_{M \times K}=\boldsymbol{F}_{M_r\times K_r}\otimes\boldsymbol{P}_{m_p\times m_p}^{\otimes L_s}$, where $\boldsymbol{P}^{\otimes L_s}$ denotes the $L_s$-times Kronecker product of $\boldsymbol{P}$ with itself, is given by
	\begin{align}\label{eq32}
	C_M=&\frac{1}{2M}\sum_{r_1+r_2+\dots+r_{m_p}=L_s}~\frac{L_s!}{r_1!r_2!\dots r_{m_p}!}\nonumber\\
	&\times\log_2\det\left({\boldsymbol{I}}_{M_r}+\rho\gamma_1^{r_1}\gamma_2^{r_2}\dots\gamma_{m_p}^{r_{m_p}}\boldsymbol{F}\boldsymbol{F}^T\right),
	\end{align}
where the summation is over all disjoint sets $\{r_1,r_2,\dots,r_{m_p}:r_i'\in\{0,1,\dots,L_s\}, \sum_{i'=1}^{m_p}r_{i'}=L_s\}$.
\end{corollary}

\noindent{\textbf{Remark 3 (Optimal Factor Matrices).}} When the design target is to maximize the average sum rate, the design of the \textit{optimal pattern matrix} can be formulated as the selection of an appropriate  rectangular matrix $\boldsymbol{F}^*$ (that can properly trade off between the performance and complexity) together with the $v_{l_s}^*$-th square matrices, $l_s=1,2,\dots,L_s$, from Algorithm 1 with the corresponding set of SNRs $\{\gamma^{(l_s)}_{v_{l_s}^*,1},\gamma^{(l_s)}_{v_{l_s}^*,2},\dots,\gamma^{(l_s)}_{v_{l_s}^*,m_{l_s}}\}$, such that they result in the maximum sum rate in \eqref{eq30}.

\noindent{\textbf{Remark 4 (Rate-Reliability Trade-off).}} After the detection of a \textit{properly} chosen subset of data symbols, one can successively cancel the interference on some detection equations such that higher SNR gains are achieved for the detection of some remained symbols (see Example 4). This, in addition to increasing the average sum rate in \eqref{eq30}, may also increase the error probability due to the error propagation issue in SIC detection, resulting in a trade-off between the rate and reliability. This mechanism also increases the latency because we should wait for the detection of some symbols before starting the detection of the others. 
	
	\noindent{\textbf{Example 4.}} For the system parameters in Example 1, by using \Tref{sum_rate_thm} and Example 2 the average sum rate is $C_{12}=(1/8)\log_2\left(1+(4/3)\rho\right)+(3/8)\log_2\left(1+(4/3)^2\rho\right)$. However, 
	after detecting $x_1$ and $x_5$ from the first two equations in \eqref{eq10}, we can apply SIC to form $\tilde{y}_9^{(2)}\triangleq {y}_5^{(1)}+{y}_9^{(1)}-2x_1-2x_5=4x_9+{n}_5^{(1)}+{n}_9^{(1)}$, instead of the third equation in \eqref{eq10}. Applying the same procedure to all of the four sets of equations defined according  to $\boldsymbol{P}^{(1)}_{3\times 3}$, increases the SNR gain $\gamma^{(1)}_{3}$ on  the third data symbols from $4/3$ to $2$. Therefore, the average sum rate is increased to $C_{12}=(1/12)\log_2\left(1+(4/3)\rho\right)+(1/4)\log_2\left(1+(4/3)^2\rho\right)+(1/8)\log_2\left(1+(8/3)\rho\right)+(1/24)\log_2\left(1+2\rho\right)$.\endproofempt
\subsection{Latency Analysis}\label{Sec4B}
Given our proposed fully parallel detection algorithm, we have the following theorem for the system latency.
\begin{theorem}\label{latency_thrm}
The maximum execution time $T_{\rm exe}^{\rm max}$ for the detection of $K$ data symbols in the proposed code-domain NOMA with the general pattern matrix $\boldsymbol{G}_{M \times K}$, specified in \eqref{eq28}, using the detection algorithm in Section \ref{Sec3} is as follows: 
\begin{align}\label{eq33}
&T_{\rm exe}^{\rm max}\!=\!t_a\!\!\sum_{l_s=1}^{L_s}\!(m'_{l_s}\!\!\!-\!1)\!+\!T_{\rm MUD}^{\rm noisy}(\boldsymbol{F}^{({L_r})}_{m_{L_r}\times k_{L_r}}\!,\mathcal{C}(K_{L_r-1},\mathcal{C}_0))+\nonumber\\
&\!\!\!\sum_{l_r=2}^{L_r-1}\!\!T_{\rm MUD}^{\rm noiseless}(\boldsymbol{F}^{({l_r})}\!,\mathcal{C}(K_{l_r\!-\!1},\mathcal{C}_0))\!+\!T_{\rm MUD}^{\rm noiseless}(\boldsymbol{F}^{({1})}\!,\mathcal{C}_0),\!
\end{align}
where $t_a$ is the required time for an addition/subtraction,
 $\mathcal{C}(K_{l_r},\mathcal{C}_0)$ denotes the constellation space of the sum of (at most) $K_{l_r}$ original data symbols each from a modulation with the constellation space $\mathcal{C}_0$, $T_{\rm MUD}^{\rm noisy}(\boldsymbol{F}^{({l_r})}_{m_{l_r}\times k_{l_r}},\mathcal{C})$ is the required time for a given MUD algorithm to detect the $k_{l_r}$ symbols with the constellation space $\mathcal{C}$ from a set of $m_{l_r}$ equations defined according to the pattern matrix $\boldsymbol{F}^{({l_r})}_{m_{l_r}\times k_{l_r}}$ in the presence of additive noise, and $T_{\rm MUD}^{\rm noiseless}(\boldsymbol{F}^{({l_r})}_{m_{l_r}\times k_{l_r}},\mathcal{C})$ is defined similar to $T_{\rm MUD}^{\rm noisy}(\boldsymbol{F}^{({l_r})}_{m_{l_r}\times k_{l_r}},\mathcal{C})$ in the absence of noise.
\end{theorem}
\begin{proof}
The detection algorithm processes $L_s$ square factor matrices through $L_s$ recursions. The $l'_s$-th recursion, for $l'_s=1,2,\dots,L_s$, involves combining sets of $m'_{l''_s}$ equations, where $l''_s\triangleq L_s-l'_s+1$, according to the rows of $\boldsymbol{\alpha}^{({l''_s})}_{m'_{l''_s}\times m'_{l''_s}}$ requiring at most $m'_{l''_s}-1$ additions/subtractions. Thus, assuming all such operations are done in parallel, the $l'_s$-th recursion requires at most $t_a(m'_{l''_s}-1)$ units of time. Now, after at most $t_a\sum_{l_s=1}^{L_s}(m'_{l_s}-1)$ units of time, we get sets of equations defined according to the Kronecker product of $L_r$ rectangular factor matrices, as \eqref{eq29}, requiring $L_r$ recursions. According to the analysis in Section \ref{Sec3C}, the first recursion involves processing sets of equations in the form of \eqref{eq21} each requiring at most $T_{\rm MUD}^{\rm noisy}(\boldsymbol{F}^{({L_r})}_{m_{L_r}\times k_{L_r}},\mathcal{C}(K_{L_r-1},\mathcal{C}_0))$ units of time. Furthermore, the middle recursions
require processing sets of equations in the form of \eqref{eq25} each incurring at most $T_{\rm MUD}^{\rm noiseless}(\boldsymbol{F}^{({l''_r})}_{m_{l''_r}\times k_{l''_r}},\mathcal{C}(K_{l''_r-1},\mathcal{C}_0))$ time{\footnote{Observe that one might be able to employ much simpler MUD algorithms in the absence of noise. Therefore, the current characterization which accounts for any difference on the detection latency (or complexity in \Tref{complexity_thrm}), given the presence or absence of noise, provides a more general formulation.}}, where $l''_r\triangleq L_r-l'_r+1$. Finally, the last recursion contains sets of equations defined according to \eqref{eq27} each necessitating $T_{\rm MUD}^{\rm noiseless}(\boldsymbol{F}^{({1})}_{m_{1}\times k_{1}},\mathcal{C}_0)$ time. Summing up the execution times of all $L_s+L_r$ recursions completes the proof.
\end{proof}   

\noindent 
\textbf{Remark 5.} \Tref{latency_thrm} characterizes the latency in the worst-case scenario. However, the actual execution time might be smaller. For example, the unknown variables of the middle recursions of rectangular factor matrices are obtained according to \eqref{eq24}. Therefore, if, for instance, half of the elements of $\boldsymbol{r}^{(L_r-l'_r)}_{i_{l''_r}+1}$ are zero, the constellation space is $\mathcal{C}(1/2K_{l''_r-1},\mathcal{C}_0)$.
	Hence, the $l'_r$-th recursion will require $T_{\rm MUD}^{\rm noiseless}(\boldsymbol{F}^{({l''_r})}_{m_{l''_r}\times k_{l''_r}},\mathcal{C}(1/2K_{l''_r-1},\mathcal{C}_0))$ time which is, in general, smaller than $T_{\rm MUD}^{\rm noiseless}(\boldsymbol{F}^{({l''_r})}_{m_{l''_r}\times k_{l''_r}},\mathcal{C}(K_{l''_r-1},\mathcal{C}_0))$. Similar arguments hold for all other recursions.

\noindent 
\textbf{Remark 6.} Note that \eqref{eq33} holds for $L_r>2$. It can be observed that with $L_r=1$ the last three terms of \eqref{eq33} should be replaced by $T_{\rm MUD}^{\rm noisy}(\boldsymbol{F}^{({1})}_{m_1\times k_1},\mathcal{C}_0)$. Moreover, for $L_r=2$, the third term in \eqref{eq33} becomes zero.

\noindent
\textbf{Remark 7.} As shown throughout this section, the proposed detection algorithm leads to a latency that grows logarithmically with the total number of users/REs. Thanks to the proposed parallel structure of the detection algorithm, a powerful BS can simultaneously process all the branches of the detection tree in uplink. Also, a given user with a typical processing power during the downlink phase can focus only on the detection of the branch containing its data symbol to avoid unnecessary detections.

\subsection{Detection Complexity}\label{Sec4C}

Recall that the proposed recursive approach results in a relatively small size for each factor matrix, even for a large $M$ and $K$; hence, it is highly desired to exactly characterize the maximum number of required operations. 
\begin{theorem}\label{complexity_thrm}
	The maximum number of additions/subtractions $N_{\rm add}^{\rm max}$ and multiplications $N_{\rm mul}^{\rm max}$ for the detection of $K$ data symbols in the proposed code-domain NOMA system with the general pattern matrix $\boldsymbol{G}_{M \times K}$, specified in \eqref{eq28}, using the recursive detection algorithm in Section \ref{Sec3} is as follows:
\begin{align}\label{eq34}
N_{\rm add}^{\rm max}=&M\sum_{l_s=1}^{L_s}(m'_{l_s}-1)+M_sN_{\rm add}^{\rm max}(\boldsymbol{F}_{M_{r}\times K_{r}}),\nonumber\\
N_{\rm mul}^{\rm max}=&M_sN_{\rm mul}^{\rm max}(\boldsymbol{F}_{M_{r}\times K_{r}}),
\end{align}
where $N_{\rm add}^{\rm max}(\boldsymbol{F}_{M_{r}\times K_{r}})$ ($N_{\rm mul}^{\rm max}(\boldsymbol{F}_{M_{r}\times K_{r}})$) is the maximum number of additions/subtractions (multiplications) for a NOMA system with the pattern matrix
$\boldsymbol{F}_{M_r\times K_r}\triangleq\boldsymbol{F}^{(1)}\otimes \dots \otimes \boldsymbol{F}^{({L_r})}$
 and the recursive detection algorithm in Section \ref{Sec3C}. $N_{\mathcal{X}}^{\rm max}(\boldsymbol{F}_{M_{r}\times K_{r}})$, $\mathcal{X}\in\{{\rm add},{\rm mul}\}$ is given by
	\begin{align}\label{eq35}
		&N_{\mathcal{X}}^{\rm max}(\boldsymbol{F}_{M_{r}\times K_{r}})\!=\!M_{L_r-1}N_{\mathcal{X}}^{\rm noisy}(\boldsymbol{F}^{({L_r})}_{m_{L_r}\times k_{L_r}}\!,\mathcal{C}(K_{L_r-1},\mathcal{C}_0))\nonumber\\
		&+\sum_{l'_r=2}^{L_r-1}\kappa^{\mathcal{S}}_{l'_r-1}M_{L_r-l'_r}N_{\mathcal{X}}^{\rm noiseless}(\boldsymbol{F}^{({L_r-l'_r+1})},\mathcal{C}(K_{L_r-l'_r},\mathcal{C}_0))\!\nonumber\\
		&+\kappa^{\mathcal{S}}_{L_r-1}N_{\mathcal{X}}^{\rm noiseless}(\boldsymbol{F}_{m_1\times k_1}^{({1})},\mathcal{C}_0),
	\end{align}
	where $M_{l_r}\triangleq\prod_{i=1}^{l_r}m_i$, $K_{l_r}\triangleq\prod_{i=1}^{l_r}k_i$, $\kappa^{\mathcal{S}}_{l'_r-1}\triangleq\prod_{i=L_r-l'_r+2}^{L_r}k_i$, $N_{\mathcal{X}}^{\rm noisy}(\boldsymbol{F}^{({l_r})}_{m_{l_r}\times k_{l_r}},\mathcal{C})$ is the required number of operations for a given MUD algorithm to detect the $k_{l_r}$ symbols with the constellation space $\mathcal{C}$ from a set of $m_{l_r}$ equations defined according to the pattern matrix $\boldsymbol{F}^{({l_r})}_{m_{l_r}\times k_{l_r}}$ in the presence of additive noise, and $N_{\mathcal{X}}^{\rm noiseless}(\boldsymbol{F}^{({l_r})}_{m_{l_r}\times k_{l_r}},\mathcal{C})$ is defined similarly in the absence of noise.
\end{theorem}
\begin{proof}
The proof is based on the details of the detection algorithm presented in Section \ref{Sec3} and following similar steps to the proof of \Tref{latency_thrm} while counting the maximum number of operations at each recursion.
\end{proof}

\noindent\textbf{Remark 8.} Note that \eqref{eq35} is obtained assuming $L_r>2$. For $L_r=1$, \eqref{eq35} should be replaced by $N_{\mathcal{X}}^{\rm noisy}(\boldsymbol{F}^{({1})}_{m_1\times k_1},\mathcal{C}_0)$. Moreover, the second term in \eqref{eq35} is zero for $L_r=2$.

{Note that \eqref{eq35} provides the full flexibility to characterize the overall complexity of the system with respect to any particular MUD algorithm, such as MAP, ML, AMP, belief propagation (BP), etc. (see, e.g., \cite{chen2017pattern,song2020super,chi2018practical,liu2019capacity}), applied to the sets of equations defined according to each of $L_r$ rectangular factor matrices at each recursion (see Section \ref{Sec3C}). In order to gain further insight into the complexity of the proposed detection algorithm, in the following, we consider the case where the overall pattern matrix is $\boldsymbol{G}_{M\times K}=\boldsymbol{F}_{M_r\times K_r}\otimes\boldsymbol{P}^{(1)}_{m'_1\times m'_1}\!\otimes\boldsymbol{P}^{(2)}_{m'_2\times m'_2}\otimes \dots \otimes \boldsymbol{P}^{({L_s})}_{m'_{L_s}\times m'_{L_s}}$, i.e., the Kronecker product of $L_s$ square factor matrices with a large $M_r\times K_r$ rectangular factor matrix $\boldsymbol{F}_{M_r\times K_r}$. Therefore, the overall complexity of the system is given as \eqref{eq34}. In Section \ref{Sec6}, we provide further numerical results to demonstrate that recursion over rectangular factor matrices also helps in reducing the detection complexity, as formulated in \eqref{eq35}. 

\begin{table*}[t]
	\centering
	{\caption{Computational complexity for the detection of all $K$ data symbols.}
		\label{T1}
		\begin{tabular}{M{4.2cm}||M{6cm}M{3cm}}  
			Algorithm & Maximum number of additions/subtractions & Maximum number of multiplications\\ \hline\hline  {\vspace{0.1cm}}
			SIC over $\boldsymbol{G}_{\!M\times K}$ & $\mathcal{O}\left(K^2M^3\right)$ & $\mathcal{O}\left(K^2M^3\right)$ \\\hline {\vspace{0.1cm}}
			Ours with SIC over $\boldsymbol{F}_{M_r\times K_r}$ & {\vspace{0.05cm}}  $\mathcal{O}\left(M+{K^2M^3}/{M_s^4}\right)$ & $\mathcal{O}\left({K^2M^3}/{M_s^4}\right)$\\\hline\hline {\vspace{0.1cm}}
			BP over $\boldsymbol{G}_{\!M\times K}$ {\vspace{0.05cm}} & {\vspace{0.05cm}}  $\mathcal{O}\left(T_{\rm in}d_gMK|\mathcal{C}_0|^{d_g}\log_2\!|\mathcal{C}_0|\right)$ & $\mathcal{O}\left(d_gMK|\mathcal{C}_0|^{d_g}\right)$\\\hline {\vspace{0.1cm}}
			Ours with BP over $\boldsymbol{F}_{M_r\times K_r}$ {\vspace{0.1cm}} & {\vspace{0.05cm}} $\mathcal{O}\left(M+{T_{\rm in}d_fMK|\mathcal{C}_0|^{d_f}\log_2\!|\mathcal{C}_0|}/{M_s}\right)$ & $\mathcal{O}\left(d_fMK|\mathcal{C}_0|^{d_f}/M_s\right)$\\\hline\hline
			{\vspace{0.1cm}}
			BP-IDD over $\boldsymbol{G}_{\!M\times K}$ {\vspace{0.05cm}} & {\vspace{0.05cm}}  $\mathcal{O}\left(T_{\rm in}T_{\rm out}d_gMK|\mathcal{C}_0|^{d_g}\log_2\!|\mathcal{C}_0|\right)$ & $\mathcal{O}\left(d_gMK|\mathcal{C}_0|^{d_g}\right)$\\\hline {\vspace{0.1cm}}
			Ours with BP-IDD over $\boldsymbol{F}_{M_r\times K_r}$ {\vspace{0.1cm}} & {\vspace{0.05cm}} $\mathcal{O}\left(M+{T_{\rm in}T_{\rm out}d_fMK|\mathcal{C}_0|^{d_f}\log_2\!|\mathcal{C}_0|}/{M_s}\right)$ & $\mathcal{O}\left(d_fMK|\mathcal{C}_0|^{d_f}/M_s\right)$\\\hline
	\end{tabular}}
\end{table*}
Assuming $\boldsymbol{G}_{\!M\times K}=\!\boldsymbol{F}_{\!M_r\times K_r}\otimes\boldsymbol{P}^{(1)}_{\!m'_1\times m'_1}\otimes\dots\otimes \boldsymbol{P}^{({L_s})}_{\!m'_{L_s}\times m'_{L_s}}$, we have $M_s\triangleq \prod_{l_s=1}^{L_s}m'_{l_s}$, $M=M_sM_r$, and $K=M_sK_r$. Note that in \eqref{eq34}, $M\sum_{l_s=1}^{L_s}(m'_{l_s}-1)=\mathcal{O}(M)$ since $\sum_{l_s=1}^{L_s}(m'_{l_s}-1)$ is a constant much smaller than $M_s\triangleq \prod_{l_s=1}^{L_s}m'_{l_s}$ and hence than $M$ and $K$. Table \ref{T1} then compares the complexity of our scheme, for the case of the pattern matrix considered here, with some other well-known detection algorithms, i.e., SIC, BP, and BP-based iterative detection and decoding (BP-IDD) \cite{chen2017pattern,liu2008low,ren2016advanced}. Note that the computational complexities in Table \ref{T1} are for the detection of all $K$ data symbols (per symbol of all users) using our proposed scheme or by directly applying the aforementioned algorithms to the overall pattern matrix $\boldsymbol{G}_{\!M\times K}$. In Table \ref{T1}, $|\mathcal{C}_0|$ denotes the size of the modulation constellation for the original data symbols. Note that our detection algorithm for the Kronecker product of square factor matrices results in sets of equations defined according to $\boldsymbol{F}_{\!M_r\times K_r}$ over subsets of size $K_r$ of original data symbols. Therefore, the constellation space of detections performed over both $\boldsymbol{F}_{\!M_r\times K_r}$ and $\boldsymbol{G}_{\!M\times K}$ is $\mathcal{C}_0$. Moreover, $T_{\rm in}$ and $T_{\rm out}$ in Table \ref{T1} represent the BP-IDD inner and outer iteration numbers, respectively. Additionally, $d_f$ and $d_g$ denote the maximum row weights of the pattern matrices $\boldsymbol{F}_{M_r\times K_r}$ and $\boldsymbol{G}_{\!M\times K}$, respectively. Note that $d_f\ll d_g$ since $K_r=K/M_s\ll K$. As seen, applying our recursive detection algorithm in conjunction with other MUD algorithms results in significantly lower complexities compared to directly applying an MUD algorithm. It is worth mentioning that the results in Table \ref{T1} are without applying recursion over the overall rectangular pattern matrix $\boldsymbol{F}_{\!M_r\times K_r}$. As numerically shown in Section \ref{Sec6}, applying our proposed recursive detection algorithm for the Kronecker product of rectangular factor matrices further lowers the complexity. Therefore, our proposed recursive detection algorithm not only significantly reduces the latency, thanks to fully parallel detection, but also noticeably lowers the overall detection complexity.
}

\section{Extensions to Practical Scenarios}\label{Sec5}
In this section, we discuss how the results of the paper can be applied to several other practical scenarios.
\subsection{Downlink and Uplink Fading Channels}\label{Sec5A}
Throughout the paper, we focused on the case of GMAC where unit/equal gains are assumed for the channels between the users and the BS in different REs. In practice, the channel gain $h_{mk}$ between the $k$-th user, for $k=1,2,\dots,K$, and the BS at the $m$-th RE, for $m=1,2,\dots,M$, can be different. In this subsection we discuss how the results of the paper can be readily extended to downlink and uplink fading channels.

\subsubsection{Downlink Transmission}\label{Sec5A1} According to \eqref{eq4}, the received signal vector at the $k$-th user is given by $\boldsymbol{y}_k={\rm diag}(h_{1k},h_{2k},\dots,h_{Mk})\boldsymbol{G}_{M\times K}\boldsymbol{x}+\boldsymbol{n}_k$. By dividing the $m$-th equation in $\boldsymbol{y}_k$ by $h_{mk}$ we get a new set of equations as $\tilde{\boldsymbol{y}}_k=\boldsymbol{G}_{M\times K}\boldsymbol{x}+\tilde{\boldsymbol{n}}_k$, where $\tilde{\boldsymbol{n}}_k(m)$ is Gaussian with mean zero and variance $\tilde{\sigma}^2_m\triangleq \sigma^2/h_{mk}^2$. Therefore, all the results can be extended to the case of downlink fading channels with the slight change that the noise components over different REs do not have identical variances though they are still independent.

\subsubsection{Uplink Transmission}\label{Sec5A2} Here, we assume in this paper that for a given user the channel gains over all $M$ REs are equal, i.e., $h_{mk}=h_k$, $\forall m,k$. This is a reasonable (and common) assumption, e.g., if the $M$ REs are the adjacent sub-channels of a slowly changing channel, they have almost equal gains. However, if the channel gains are different, the users can adjust their power at each RE in such a way that equal gain is observed on the received signals over different REs. Note that this assumption for the case of downlink transmission yields exactly the same channel model $\boldsymbol{y}_k/h_k=\boldsymbol{G}_{M\times K}\boldsymbol{x}+\boldsymbol{n}_k$ as GMAC  with identical noise variances over different REs.

Recall, based on \eqref{eq3}, that the uplink signal received by the BS can be expressed as $\boldsymbol{y}=\boldsymbol{H}\boldsymbol{x}+\boldsymbol{n}$, where $\boldsymbol{H}\triangleq\boldsymbol{\mathcal{H}}{\odot}\boldsymbol{G}_{M\times K}$ and $\boldsymbol{\mathcal{H}}\triangleq[\boldsymbol{h}_1,\boldsymbol{h}_2,\dots,\boldsymbol{h}_K]$. With the above assumption, we have $\boldsymbol{h}_k=h_k\boldsymbol{1}_{M\times 1}$, $\forall k$. The received signal vector can then be rewritten as $\boldsymbol{y}=\boldsymbol{G}_{M\times K}\tilde{\boldsymbol{x}}+\boldsymbol{n}$, where $\tilde{x}_k\triangleq \tilde{\boldsymbol{x}}(k)=h_kx_k$. Therefore, all of our earlier results can be readily extended to the case of uplink fading channels with the slight change that each $k$-th data symbol is scaled by the channel gain $h_k$.

{\noindent\textbf{Remark 9.} As clarified in this section, one can readily extend the results of the paper from the case of GMAC model to fading channels. One of the important extensions is the extension of Algorithm 1 to the case of fading channels. The discussions in this section demonstrate that, in the case of downlink fading channels, given a square factor matrix $\boldsymbol{P}$, the corresponding optimal combining matrix $\boldsymbol{\alpha}$ and SNR gains are dependent to the statistical properties of the channel and should update every time the downlink channel gains change. This is because for the case of downlink fading channels the noise components have different variances (inversely proportional with the square of the channel gains); hence, we need to revise the definition of the SNR gains with respect to the noise variances which intuitively suggests assigning higher weights to the stronger channels to maximize the SNR gains. 
	  Therefore, it is more desirable to apply the proposed algorithm in this paper to slow fading downlink channels. Note that this is also the case with other code-domain NOMA schemes where the overall pattern matrix should be updated after each change of the channel matrix. However, as discussed in Section \ref{Sec3A}, the complexity of such updates is much smaller here due to the factorization of the overall pattern matrix which significantly reduces the search space for our proposed algorithm compared to conventional PDMA algorithms. Additionally, in the case of downlink fading channels, the BS potentially has the computing power to run Algorithm 1 every time the downlink channel matrix changes.
  On the other hand, in the case of uplink fading channels, we have $\boldsymbol{y}=\boldsymbol{G}_{M\times K}\tilde{\boldsymbol{x}}+\boldsymbol{n}$. Although this changes the SNRs of different users according to their channel gains, it does not require Algorithm 1 to run again. Therefore, Algorithm 1, in its current form, can be applied to find the corresponding combining matrices and the SNR gains, and one can use the pre-stored results from Algorithm 1 for different factor matrices.
}

\subsection{Joint Power- and Code-Domain NOMA (JPC-NOMA)}\label{Sec5B}

The code-domain NOMA scheme proposed in this paper can be combined in many ways with power-domain NOMA. For instance, if, for some $l$, $2^{m_l}-1<k_l$, then we cannot satisfy distinct nonzero columns for the factor matrix $\boldsymbol{G}^{(l)}_{m_l\times k_l}$. Consequently, the overall pattern matrix $\boldsymbol{G}_{M\times K}$ will have some repeated columns, i.e., same pattern vectors for some of the users (see, e.g., \cite[Example 1]{jamali2018low}).
Therefore, one cannot distinguish between the data of the users with an equal pattern vector. In other words, the data symbols of the users with an equal pattern vector will be paired together, i.e., they coexist in all sets of equations containing any of those symbols.
In such a case, we propose to assign different power coefficients to the users having the same pattern vector in order to differentiate between them using power-domain NOMA techniques. Then, the receiver after detecting the paired symbols using our proposed detection algorithm, detects the individual data symbols contained in each paired symbol according to the principles of power-domain NOMA.
In this case, users with the most disparate channel qualities should be paired together.

More importantly, we can incorporate power-domain NOMA principles to increase the reliability of the detection performed over sets of equations at each layer of the detection algorithm proposed in Section \ref{Sec3C}.
 Recall that each of those sets of equations contains a disjoint subset of symbols.
  For example, considering the general case of pattern matrix, specified in \eqref{eq28}, the final recursion involves $M_s\prod_{l_1=2}^{L_r}k_{l_r}$ sets of equations defined similar to \eqref{eq27}, each containing $k_1$ disjoint symbols while the indices of the symbols are known. As shown in \cite{chen2017pattern}, assigning different power scaling and phase shifting factors to the data symbols in a PDMA system improves the detection reliability. Therefore, the transmitter can group each subset of $k_1$ users
  together. Then different power levels (as well as, possibly, phase shifts) are assigned to the data symbols of the users in each group according to the principles of power-domain NOMA in order to increase the reliability of the detection over the aforementioned $M_s\prod_{l_1=2}^{L_r}k_{l_r}$ sets of equations.

\section{Numerical Results}\label{Sec6}
In this section, we numerically compare the average sum-rate of various configurations, including the traditional OMA, regular PDMA with the reported optimal pattern matrices in \cite{chen2017pattern} such as $3\times 6$ and $4\times 8$ matrices with the maximum row weight $d_f=4$, and our proposed recursive code-domain NOMA with the square factor matrices obtained from Algorithm 1.
{ For our numerical analysis in this section, we consider the GMAC model that corresponds to the case where all channel gains are equal to one. Additionally, we assume that each user transmits with power $P_x$, and the noise variance of the receiver at each RE is $\sigma^2$.}
By running Algorithm 1 for $m_l=3$ and $4$, we get the optimal square factor matrices $\boldsymbol{P}^{(1)}_{3\times 3}$ and $\boldsymbol{P}^{(2)}_{4\times 4}$ in \eqref{eq7} with the combining matrices as \eqref{eq8_alpha12} and individual SNR gains reported in Example 2. 
\begin{figure}[t]
	\centering
	\includegraphics[trim=0.7cm 0.2cm 0 0,width=3.7in]{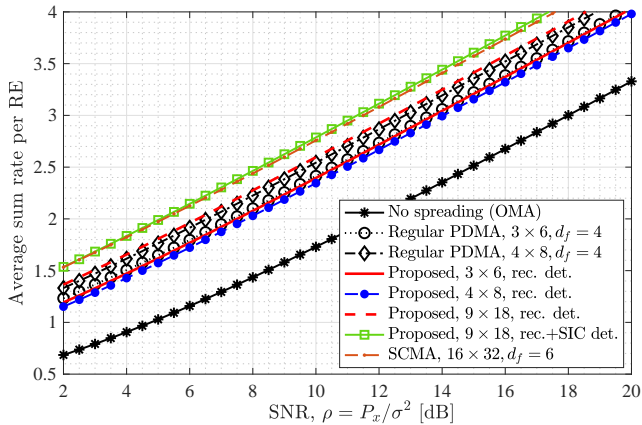}
	\caption{Average sum-rate of various multiple access mechanisms including the traditional OMA, regular PDMA, and our proposed recursive method.}
	\label{fig3}
	\vspace{-0.1in}
\end{figure}

Fig. \ref{fig3} shows the per-RE average sum-rate of various multiple access techniques. Here, for our proposed scheme, we considered $\boldsymbol{G}_{M\times K}=\boldsymbol{F}\otimes \boldsymbol{P}^{\otimes r}$ with $\boldsymbol{F}=[1~1]$, $\boldsymbol{P}$ being either $\boldsymbol{P}^{(1)}_{3\times 3}$ or $\boldsymbol{P}^{(2)}_{4\times 4}$, and $r=1$ or $2$.
{It is observed that both the regular PDMA and our proposed scheme significantly outperform the traditional OMA method. Moreover, for relatively similar system parameters, our proposed method with the recursive detection performs very close to the capacity of regular PDMA with the optimal pattern matrix $\boldsymbol{A}_{M\times K}$ and optimal MAP detection, i.e., $C^{\rm PDMA}_M=\frac{1}{2M}\log_2\det\left({\boldsymbol{I}}_{M}+\rho\boldsymbol{A}\boldsymbol{A}^T\right)$.
Note that the negligible loss in the sum-rate performance of our scheme compared to that of PDMA, for a given dimension $M$ and $K$ of the pattern matrix, is because of two major reasons. First, for PDMA we assumed optimal MAP detection over the overall pattern matrix and plotted the capacity formula given the pattern matrix. However, for our scheme we are applying our low-complexity detection scheme which has a much lower complexity than MAP. Second, the design of the pattern matrix for PDMA is more complex. Indeed, it searches for the best $M\times K$ pattern matrix to maximize the sum rate. Our scheme, on the other hand, designs the pattern matrix as the Kronecker product of smaller factor matrices such that each factor matrix has some properties to render a low-complexity detection. On the positive side, due to significantly lower complexity of our scheme, we can utilize pattern matrices with higher dimensions (e.g., by increasing $r$) and, also, apply SIC detection at some recursions similar to Example 4 (see also \cite[Example 4]{jamali2018low}) to get larger sum rates. This can significantly boost the performance of our scheme for a fixed overload factor.}

{Fig. \ref{fig3} also shows the capacity formula curve (i.e., the best performance one can achieve) for the SCMA method having a $16\times32$ pattern matrix $\boldsymbol{H}_{16\times32}$ with a row weight of 6 and a column weight of 3. This sparse pattern matrix is formed by picking the parity-check matrix of a regular low-density parity-check (LDPC) code according to \cite{mansour2006640}. It is observed that the performance of the SCMA method with this pattern matrix is almost the same as the green curve, i.e., our method with a $9\times18$ pattern matrix decoded recursively (with SIC applied at some recursions), according to our proposed detection algorithm. This further highlights the potentials of our proposed scheme given that it is capable of achieving almost the same performance as SCMA with a twice smaller dimension and a low-complexity decoder. Additionally, the performance of our proposed scheme  is significantly improved by increasing the dimensions of the overall pattern matrix (as shown in Fig. \ref{fig4}) while the SCMA performance does not change much by increasing the pattern matrix dimension (given the sparsity constraint on the pattern matrix which does not allow transmitting the symbols of the users over many REs).}

\begin{figure}[t]
	\centering
	\includegraphics[trim=0.7cm 0.2cm 0 0,width=3.7in]{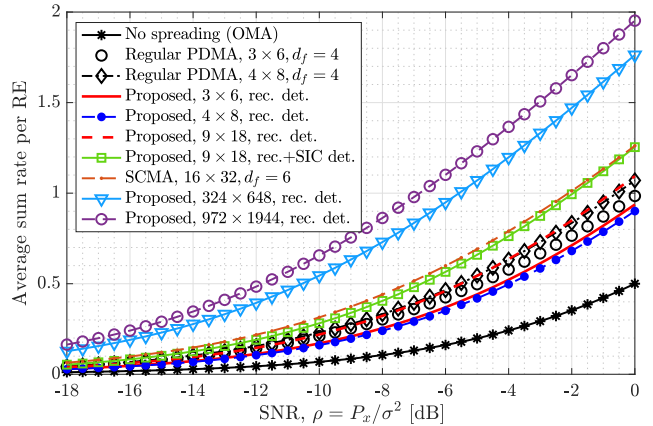}
	\caption{Average sum-rate of various multiple access schemes at low SNRs.}
	\label{fig4}
	\vspace{-0.1in}
\end{figure}
{In Fig. \ref{fig4}, the focus is specifically on low-SNR regimes. Here, we consider $\boldsymbol{G}_{M\times K}=\boldsymbol{F}_{4\times8}\otimes {\boldsymbol{P}^{(1)}_{3\times 3}}^{\!\!\!\otimes r}$ for the proposed scheme with $\boldsymbol{F}_{4\times8}$ being the $4\times 8$ PDMA matrix with $d_f=4$ from \cite{chen2017pattern}. In this case, $M=4\times3^r$, $K=8\times3^r$, and the overload factor $\beta=2$. It can be observed in Fig. \ref{fig4} that having large pattern matrices are essential at low SNRs to avoid getting close-to-zero sum rates. While it is not practically feasible to implement conventional code-domain NOMA schemes at such large dimensions considered in this setting, our low-complexity recursive approach makes it possible to spread the data symbols of the users using desirably large pattern vectors. As a consequence, $2.74$  and $3.46$ times larger average sum rates are achieved for our recursive scheme with $\boldsymbol{G}_{324\times 648}$ (i.e., $r=4$) and $\boldsymbol{G}_{972\times 1944}$ (i.e., $r=5$), respectively, compared to $4\times 8$ PDMA scheme at the SNR of $-15$ \si{dB}. These gains are as large as $10.25\times$ and $12.91\times$, respectively, when compared to OMA. One may deduce that such gains are directly as a result of the development of the SNR gains through the proposed recursive detection algorithm in Section \ref{Sec3B2}. Indeed, using \Cref{sum_rate_cor}, the average sum-rate per RE for the pattern matrix design of $\boldsymbol{G}_{M\times K}=\boldsymbol{F}_{4\times8}\otimes {\boldsymbol{P}^{(1)}_{3\times 3}}^{\!\!\!\otimes r}$ is $C_M=\frac{1}{8}\log_2\det\left({\boldsymbol{I}}_{4}+\rho(4/3)^{r}\boldsymbol{F}_{4\times8}\boldsymbol{F}^T_{4\times8}\right)$. Therefore, starting from $r=0$, each unit increase of $r$ improves the average sum-rate by $10\log_{10}(4/3)\approx1.25$ \si{dB}.}

\begin{figure}[t]
	\centering
	\includegraphics[trim=0.7cm 0.2cm 0 0,width=3.4in]{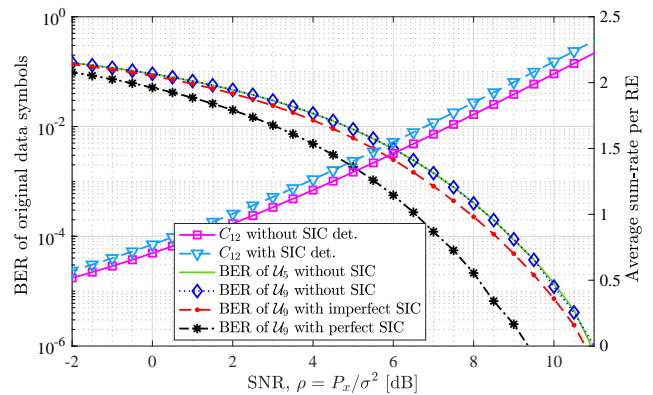}
	\caption{Average sum-rate and BER comparison for the detection with and without SIC.}
	\label{fig5}
	\vspace{-0.1in}
\end{figure}
{ In order to analyze the bit error rate (BER) and also observe the impact of SIC detection on the BER, we consider the setup in Example 1. The average sum-rate $C_{12}$ of the system is characterized in Example 4. Additionally, we have considered a sample way of incorporating SIC detection in the middle of the detection process, in Example 4. The increase on the average sum-rate by incorporating SIC detection is evident from Example 4 and is also confirmed in Fig. \ref{fig5}. For the BER performance we assume that all users employ binary phase shift keying (BPSK) modulation and, similar to Example 1, we only focus on the detection of the users $\mathcal{U}_1$, $\mathcal{U}_5$, and $\mathcal{U}_9$.
Using Monte-Carlo simulations with $10^7$ trials, it is confirmed that all three users have the same BER performance (without SIC detection) since the overall SNR gain for all of them is $\gamma_{t,i}=(4/3)^2$, $i=1,5,9$ (see Example 2). Additionally, as shown in Example 4, incorporating the SIC detection, according to Example 4, increases the overall SNR gain of $\mathcal{U}_9$ from $(4/3)^2$ to $8/3$. This is equivalent to $10\log_{10}(3/2)\approx1.76$ \si{dB} gain on the BER performance of $\mathcal{U}_9$, assuming perfect SIC detection, which is also verified in Fig. \ref{fig5}. 
{However, due to the error propagation on the SIC detection, the gain on the BER of $\mathcal{U}_9$ with SIC detection is much smaller than $1.76$ \si{dB}, demonstrating the impact of imperfect SIC detection. In particular, there is around $1.3-1.5$ \si{dB} gap between perfect and imperfect SIC detection over almost all ranges of SNR.}
Nevertheless, interestingly, the performance of SIC detection (according to the sample way of incorporation as explained in Example 4) improves both average sum-rate and the BER performance. This is mainly because we only employed SIC detection over the last step of the detection, and once the detection of the earlier symbols is reliable, the gain of the SNR is dominant to the imperfection of the SIC detection. In general, the the BER performance may degrade if we incorporate SIC detection in many middle steps, calling for an interesting trade-off between the rate and reliability (see also Remark 4).}

{In Section \ref{Sec4C}, by considering the overall pattern matrix as $\boldsymbol{G}_{\!M\times K}=\!\boldsymbol{F}_{\!M_r\times K_r}\otimes\boldsymbol{P}^{(1)}_{\!m'_1\times m'_1}\otimes\dots\otimes \boldsymbol{P}^{({L_s})}_{\!m'_{L_s}\times m'_{L_s}}$, we have shown that applying our recursive detection algorithm in conjunction with other MUD algorithms significantly lowers the detection complexity compared to directly applying an MUD algorithm. Here, we provide a numerical example to show that applying our proposed recursive detection algorithm for the Kronecker product of rectangular factor matrices further lowers the complexity. To see this, let us consider the overall pattern matrix as $\boldsymbol{F}_{2\times 6}=\boldsymbol{F}^{(1)}_{1\times 2}\otimes\boldsymbol{F}^{(2)}_{2\times 3}$,  where $\boldsymbol{F}^{(1)}_{1\times 2}$ and $\boldsymbol{F}^{(2)}_{2\times 3}$ are given in Example 3. Let us further assume that we want to apply MAP for the detection over the system of equations defined as $\boldsymbol{y}_{2\times 1}=\boldsymbol{F}_{2\times 6}\boldsymbol{x}_{6\times 1}+\boldsymbol{n}_{2\times 1}$, and we consider the number of possibilities that a MAP detector has to search for as a rough approximate of the detection complexity. In this case, directly applying MAP to this system of equations requires searching over all $|\mathcal{C}_0|^6$ possibilities. On the other hand, applying our recursive detection algorithm requires $|\mathcal{C}(2,\mathcal{C}_0)|^3+3|\mathcal{C}_0|^2$ searches, where $\mathcal{C}(2,\mathcal{C}_0)$ is the constellation space of the sum of $2$ original data symbols each from a modulation with the constellation space $\mathcal{C}_0$ (see \eqref{eq35} for the details). Assuming quadrature phase shift keying (QPSK) modulation as an example, $|\mathcal{C}_0|=4$ and $|\mathcal{C}(2,\mathcal{C}_0)|=9$. Therefore, our scheme requires computations of $777$ cases which is much smaller than $4^6=4096$, i.e., direct application of MAP to $\boldsymbol{F}_{2\times 6}\boldsymbol{x}_{6\times 1}$.}

\section{Conclusions and Future Directions}\label{Sec7}
{We proposed a low-complexity and scalable recursive approach toward code-domain NOMA by constructing the overall pattern matrix as the Kronecker product of several factor matrices. We then developed detection algorithms that reduce to several layers/recursions,  each dealing with disjoint subsets of equations corresponding to certain factor matrices, that can be executed in parallel.
For the Kronecker product of square factor matrices we proposed a systematic way of choosing the factor matrices that enables a remarkably low-complexity detection involving only few linear operations (additions/subtractions). 
Furthermore, for the Kronecker product of rectangular factor matrices we provided a recursive detection algorithm that can work on the general case of factor matrices. Given the proposed schemes and the detection algorithm we characterized the system performance in terms of average sum rate, latency, and detection complexity. We further discussed possible extensions of the work to the case of fading channels and joint power- and code-domain NOMA.
We showed that the proposed scheme has significantly lower complexity and latency compared to straightforward code-domain NOMA schemes. Moreover, it is numerically verified that by utilizing large pattern matrices the proposed scheme significantly improves the average sum rate.
}

{The proposed approach in this paper can be extended in various directions. Here, we highlight several directions for the future research.
\subsubsection{Optimal design of factor matrices} \Tref{sum_rate_thm} and \Cref{sum_rate_cor} naturally lead to a design strategy for optimal factor matrices that maximize the average sum rate (see Remark 3). The characterization of other performance metrics, such as individual rates \cite{xu2015new}, and then finding the optimal factor matrices with respect to these metrics is a direction for future research.
\subsubsection{Rate-Reliability-Latency Trade-offs} Given the set of factor matrices and the overall recursive approach discussed in Section \ref{Sec3}, various detection schemes can be applied at each layer/recursion trading off between the rate, reliability, and latency, as exemplified in Remark 4 and Example 4. Detailed characterizations of such trade-offs is another direction for future research.
\subsubsection{Joint Power- and Code-Domain NOMA (JPC-NOMA)} As discussed in Section \ref{Sec5B}, power-domain NOMA can be combined with the proposed code-domain NOMA in several ways to improve the system design in terms of various performance metrics including reliability, throughput, fairness, etc. Designing such joint networks and analytically characterizing their performance metrics to show the advantages of the joint design is another important future research direction.
\subsubsection{Grant-Free Transmission} In many use cases of massive communication, it is essential to reduce the coordination overhead of the communication protocol. Grant-free protocols attempt to serve a massive number of users with minimal (or even without any) coordination. Incorporating the proposed low-complexity code-domain NOMA scheme in the context of grant-free transmission is a vital future direction.}

\appendices
\section{Proof of \Lref{lem1}}\label{AppA}
Let $\boldsymbol{P}=\left[p_{i,j}\right]_{m\times m}$. 
If there exists no such a combining matrix $\boldsymbol{\alpha}=\left[\alpha_{i,j}\right]_{m\times m}$ with $\alpha_{i,j}\in\{-1,0,+1\}$, then we are done. Otherwise, if such a matrix exists, then we prove the uniqueness of the $i$-th row of $\boldsymbol{\alpha}$, for $i=1,2,\dots,m$, by contradiction. 

Assume to the contrary that there are $T=2$ distinct sets of coefficients  $\{\alpha^{(1)}_{i,j}\}_{j=1}^m$ and $\{\alpha^{(2)}_{i,j}\}_{j=1}^m$ for the $i$-th row of $\boldsymbol{\alpha}$ that result in singleton vectors with nonzero elements at the $i$-th position of the $i$-th row of $\boldsymbol{\alpha}\boldsymbol{P}$. In other words, we have
	\begin{align}
		\sum_{j=1}^{m}\alpha^{(t)}_{i,j}p_{j,i}&=w_t\neq 0,\\
		\sum_{j=1}^{m}\alpha^{(t)}_{i,j}p_{j,i'}&=0, ~~~~ i'\neq i=1,2,\dots,m,
	\end{align}
for $t=1,2$, where $w_1$ and $w_2$ are two nonzero integers. Then we can define a new set of combining coefficients $\{\alpha^{(3)}_{i,j}\}_{j=1}^m$ with $\alpha^{(3)}_{i,j}\triangleq w_1\alpha^{(2)}_{i,j}-w_2\alpha^{(1)}_{i,j}$. Note that
\begin{align}
	\sum_{j=1}^{m}\alpha^{(3)}_{i,j}p_{j,i}&=w_1\sum_{j=1}^{m}\alpha^{(2)}_{i,j}p_{j,i}-w_2\sum_{j=1}^{m}\alpha^{(1)}_{i,j}p_{j,i}\nonumber\\
	&=w_1w_2-w_2w_1=0,\label{apA_3}\\
	\sum_{j=1}^{m}\alpha^{(3)}_{i,j}p_{j,i'}&=w_1\sum_{j=1}^{m}\alpha^{(2)}_{i,j}p_{j,i'}-w_2\sum_{j=1}^{m}\alpha^{(1)}_{i,j}p_{j,i'}\nonumber\\
	&=0-0=0, ~~~~ i'\neq i=1,2,\dots,m.\label{apA_4}
\end{align}
Now, two cases are possible:
\begin{itemize}
	\item $\alpha^{(3)}_{i,j}=0$, i.e., $\alpha^{(2)}_{i,j}=(w_2/w_1)\times \alpha^{(1)}_{i,j}$ for all $j=1,2,\dots,m$. This implies that there is only one set of coefficients $\{\alpha^{(1)}_{i,j}\}_{j=1}^m$ for the $i$-h row of $\boldsymbol{\alpha}$ since scaling the combining coefficients by a constant factor (here, $w_2/w_1$) does not change the performance and the SNR gains as explained for the condition \textbf{C1} in Section III-B1. 
	\item  At least for one $j$, $\alpha^{(3)}_{i,j}\neq0$. This is in contradiction with the linear independence of the rows of $\boldsymbol{P}$ since the linear combination of the rows with nonzero combining coefficients is equal to zero as shown by \eqref{apA_3} and \eqref{apA_4}.
\end{itemize}
This proves that the $i$-th row of $\boldsymbol{\alpha}$ is unique. Repeating the same procedure for $i=1,2,\dots, m$ completes the proof. \endproof

\bibliographystyle{IEEEtran}
\bibliography{IEEEabrv}

\end{document}